\documentclass[pdflatex,sn-mathphys-num]{sn-jnl}

\usepackage{graphicx}%
\usepackage{multirow}%

\usepackage{amssymb}
\usepackage[title]{appendix}%
\usepackage{xcolor}%
\usepackage{textcomp}%
\usepackage{manyfoot}%
\usepackage{booktabs}%
\usepackage{algorithm}%
\usepackage{algorithmicx}%
\usepackage{algpseudocode}%

\usepackage{listings}%

\usepackage[T1]{fontenc}
\usepackage{thm-restate}
\usepackage{enumitem}

\usepackage{ulem}

\theoremstyle{thmstyleone}%
\newtheorem{theorem}{Theorem}
\newtheorem{lemma}[theorem]{Lemma}
\newtheorem{corollary}[theorem]{Corollary}
\newtheorem{claim}[theorem]{Claim}

\theoremstyle{thmstyletwo}%

\theoremstyle{thmstylethree}%
\newtheorem{definition}{Definition}%

\newcommand{\polylog}{\mathrm{polylog}}
\newcommand{\dm}{\mathrm{dm}}
\newcommand{\ap}{\mathit{ap}}

\renewcommand{\leq}{\leqslant}
\renewcommand{\geq}{\geqslant}

\begin{document}

\title[An Improved Pseudopolynomial Time Algorithm for Subset Sum]{An Improved Pseudopolynomial Time Algorithm for Subset Sum\footnote{A preliminary version of this paper appeared in Proceedings of FOCS 2024.}}


\author[1]{\fnm{Lin} \sur{Chen}}\email{chenlin198662@zju.edu.cn}

\author*[1]{\fnm{Jiayi} \sur{Lian}}\email{jiayilian@zju.edu.cn}

\author[1]{\fnm{Yuchen} \sur{Mao}}\email{maoyc@zju.edu.cn}

\author[1]{\fnm{Guochuan} \sur{Zhang}}\email{zgc@zju.edu.cn}

\affil[1]{\orgdiv{College of Computer Science and Technology}, \orgname{Zhejiang University}, \country{China}}


\abstract{We investigate pseudo-polynomial time algorithms for Subset Sum. Given a multi-set $X$ consisting of $n$ positive integers and a target $t$, Subset Sum asks whether some subset of $X$ sums to $t$. Bringmann proposed an $\widetilde{O}(n + t)$-time algorithm [Bringmann SODA'17]. An open question has naturally arisen: can Subset Sum be solved in $O(n + w)$ time? Here $w$ is the largest integer in $X$. We make progress towards resolving the open question by proposing an $\widetilde{O}(n + \sqrt{wt})$-time algorithm.}




\maketitle

\section{Introduction}
Given a target $t$ and a multi-set $X$ consisting of $n$ positive integers, Subset Sum asks whether some subset of $X$ sums to $t$. Subset Sum is among Karp's 21 NP-complete problems~\cite{Kar72}, and serves as a hard special case of many other problems, e.g., Knapsack, Integer Programming, and Constrained Shortest Path.

It is well known that Subset Sum admits pseudo-polynomial time algorithms: the textbook dynamic programming~\cite{Bel57} solves Subset Sum in $O(nt)$ time. In recent years, there has been growing interest in searching for more efficient algorithms for Subset Sum. Pisinger presented an $O(nt/\log t)$-time algorithm~\cite{Pis03}, as well as an $O(nw)$-time algorithm~\cite{Pis99}, where $w$ refers to the largest integer in $X$. Koiliaris and Xu~\cite{KX19} showed an $\widetilde{O}(\sqrt{n}t)$-time algorithm\footnote{$\widetilde{O}(T)$ hides polylogarithmic factors in $T$.}, which was later simplified by the same authors~\cite{KX18}.  Bringmann~\cite{Bri17} proposed an $\widetilde{O}(n+t)$-time randomized algorithm that solves Subset Sum with high probability. Later, Jin and Wu~\cite{JW19} presented an alternative $\widetilde{O}(n+t)$ time-randomized algorithm. 
Meanwhile, there is a conditional lower bound of $(n+w)^{1-o(1)}$ implied by the Set Cover Hypothesis~\cite{Bri17}
and the Strong Exponential Time Hypothesis~\cite{ABHS22}. An important open question repeatedly mentioned in a series of papers~\cite{ABJ+19, BW21, PRW21, BC22, Jin24} asks whether Subset Sum can be solved in $O(n+w)$ time. Several attempts have been made to resolve the open problem. Polak, Rohwedder and W{\k{e}}grzycki \cite{PRW21} showed that Subset Sum can be solved in $\widetilde{O}(n+w^{5/3})$ time, and it was later improved to $\widetilde{O}(n+w^{3/2})$ time by Chen, Lian, Mao and Zhang~\cite{CLMZ24aSODA} and Jin~\cite{Jin23}, independently. These algorithms outperform the $\widetilde{O}(n+t)$-time algorithm only in the regime $t \gg w^{3/2}$.\footnote{$\gg$ hides polylogarithmic or sub-polynomial factors.}

We remark that the complexity of Subset Sum can be fairly complicated when we further consider the relationship among the parameters $n,t,w$ (see, e.g.~\cite{Cha99,CFG89,GM91,BW21}). In particular, when the input $X$ is a set (i.e., $X$ has no repeated integers), Bringmann and Wellnitz~\cite{BW21} showed that the instance can be solved in $\widetilde{O}(n)$ time if $w\sum_ix_i/n^2\ll t\leq \sum_ix_i/2$, which implies that $n\gg \sqrt{w}$. Such a case is referred to as ``Dense Subset Sum''. They also obtained a similar (but weaker) result for multi-sets. Note that the algorithm does not work for the region $n\leq \sqrt{w}$, and thus it does not contradict the conditional lower bound of $(n+w)^{1-o(1)}$.

\subsection{Our Result}
Our main result is the following.
\begin{restatable}{theorem}{thmmain}
 \label{thm:main}
        Subset Sum can be solved in $\widetilde{O}(n +\sqrt{wt})$ time by a randomized, one-sided-error algorithm with success probability $1 - (n + t)^{-\Omega(1)}$.
\end{restatable}
More precisely, if there is a subset that sums to $t$, our algorithm returns yes with probability $1 - (n + t)^{-\Omega(1)}$. Otherwise, it returns no.

By removing all integers greater than $t$, we can, without loss of generality, assume that $w\leq t$. Our algorithm gives the first improvement over the $\widetilde{O}(n + t)$-time algorithm~\cite{Bri17} within the whole regime $t \geq w$. 

We remark that the algorithm we presented in this paper only solves the decision version of
subset sum. That is, it can only decide whether there is a solution, but constructing a solution would require much more time. This is mainly because the additive combinatorics result used by our algorithm is non-constructive. This additive combinatorics result (Theorem~\ref{thm:ap}) has been made constructive by Chen, Mao, and Zhang~\cite{CMZ26} very recently. Combining our algorithm with their result, a solution can also be found in $\widetilde{O}(n +\sqrt{wt})$ time with high probability.




Combining our algorithm with the result for Dense Subset Sum~\cite{BW21}, we can show that, when the input is a set (rather than a multi-set), Subset Sum can be solved in $\widetilde{O}(n + w^{1.25})$ time, improving upon the previous best running time of $\widetilde{O}(n + w^{1.5})$~\cite{CLMZ24aSODA,Jin23}.
\begin{corollary}\label{coro:1}
    When the input is a set, Subset Sum can be solved in $\widetilde{O}(n + w^{1.25})$ time by a randomized, one-sided-error algorithm with probability $1 - (n + t)^{-\Omega(1)}$.
\end{corollary}
We give a brief argument. When $n = {O}(\sqrt{w})$, since $t \leq nw \leq {O}(w^{1.5})$, our $\widetilde{O}(n + \sqrt{wt})$-time algorithm will run in $\widetilde{O}(n + w^{1.25})$ time. Now suppose that $n = {\Omega}(\sqrt{w})$. If $t = {O}(w^2\log n/n)$, again, the $\widetilde{O}(n + \sqrt{wt})$-time algorithm will run in $\widetilde{O}(n + w^{1.25})$ time. It remains to consider the case where $n = {\Omega}(\sqrt{w})$ and $t = {\Omega}(w^2\log n/n)$. Since $\sum_ix_i\leq nw$, we have $w\sum_ix_i/n^2 \leq w^2/n\ll t$. Thus, this case can be solved in $\widetilde{O}(n)$ time according to the results for dense Subset Sum~\cite{BW21}. Note that this argument only works for sets. When the input is a multi-set, Dense Subset Sum is not strong enough to make the above argument valid.

\subsection{Technical Overview}
As in many previous works (e.g. \cite{CLMZ24cSTOCPartition,BW21,DJM23}), we reduce Subset Sum to a more general problem of computing (a subset of) the collection $\mathcal{S}(X)$ of all possible subset sums. 

Very recently, Chen, Lian, Mao, and Zhang~\cite{CLMZ24cSTOCPartition} proposed an efficient framework for (approximately) computing $\mathcal{S}(X)$, and obtained an $\widetilde{O}(n+1/\varepsilon)$-time FPTAS (fully polynomial time approximation scheme) that weakly approximates Subset Sum. Algorithm~\ref{alg:sketch} gives a sketch of their framework. It computes $\mathcal{S}_X$ in a tree-like manner: it iteratively constructs a new level of the tree by computing the sumset of every two sets in the previous level. Using sparse convolution algorithms, each level can be constructed in time roughly linear in the total size of the sets in this level. Therefore, if every level of the tree has a small total size, $\mathcal{S}(X)$ can be obtained efficiently.  When some level has a large total size, it can be proved that $\mathcal{S}(X)$ contains an arithmetic progression $a_1, a_2, \ldots, a_k$ with common difference at most $1/\varepsilon$ and $a_1 \leq t \leq a_k$. This immediately implies that some subset sum of $X$ is close to $t$, and therefore, gives a good approximation. The key of this argument is the following additive combinatorics result, which is implied by the work of Szemer{\'e}di and Vu~\cite{SV05}: given $\ell$ sets $A_1,A_2,\cdots,A_{\ell}$ of positive integers not exceeding $a_{\max}$, if their total size $\sum_{i=1}^\ell|A_i|$ is large compared to $a_{\max}$, i.e., $\sum_{i=1}^\ell|A_i|\geq \ell+Ca_{\max}\log a_{\max}$ for some constant $C$, then the sumset $A_1+A_2+\cdots+A_{\ell}$ must contain an arithmetic progression of length at least $a_{\max}$. 

\begin{algorithm}[ht]
    \caption{\texttt{Dense-or-Sparse Framework} (A Sketch)}
    \label{alg:sketch}
    \begin{algorithmic}[1]
    \Statex \textbf{Input:} a multi-set $X$ of nonnegative integers
    \State $S^0_i = \{x_i, 0\}$ for $x_i \in X$
    \For{$j := 1,..., \log n$}
        \State Define $S^j_i := S^{j-1}_{2i-1} + S^{j-1}_{2i}$
        \State Let $u$ be the maximum integer in $S^j_1, S^j_2, S^j_3, \ldots$
        \If{$\sum_i |S^j_i| \gg u$ } 
            \State dense case: stop the computation and apply additive combinatorics tools.
        \Else 
            \State sparse case: compute all sets $S^j_i$ using sparse convolution algorithms
        \EndIf
    \EndFor
    \State \Return $S^{\log n}_1$
    \end{algorithmic}
\end{algorithm}

Adopting the above framework in exact algorithms turns out to be much more challenging. Technical issues that can be resolved easily in approximation setting now become difficult. Below, we briefly explain those issues and our techniques to resolve them.

\paragraph{Maintaining a small threshold for being dense.} 
As in Algorithm~\ref{alg:sketch}, the threshold for a level being dense is roughly $u$, which is the largest integer of the sets in this level.  When a level is just sparse, we may need as much as $\widetilde{O}(u)$ time to construct this level. The maximum integer grows as we construct new levels, and it becomes $\Theta(nw)$ at the last level. A direct application of the framework will require $\Theta(nw)$ time.

In approximation algorithms, one can force $u$ to be $\widetilde{O}(1/\varepsilon)$ by scaling, and therefore, reduce the running time to $\widetilde{O}(1/\varepsilon)$. Scaling, however, is not allowed in exact algorithms. To address this issue, we utilize the technique of random permutation. The crucial observation is that, when we randomly permute the input integers and feed them to the dense-or-sparse framework, the contribution of each set $S^j_i$ to $t$ is close to its expectation. More precisely, with high probability, the contribution of $S^j_i$ to $t$ is within an interval of length $2\sqrt{wt}$ around the expectation. Hence, for each $S^j_i$, there is no need to consider the entire set; it suffices to consider the part within this interval. This method of random permutation (partition) has been used in several recent pseudo-polynomial time algorithms for the knapsack problem (see, e.g.~\cite{BHSS18,BC23,DJM23,HX24,BDP24}).
    
\paragraph{Using a long arithmetic progression.} 
    
In approximation algorithms, the existence of an arithmetic progression with a small common difference immediately gives a good approximation. Such an arithmetic progression is insufficient if we need an exact answer to whether $t \in \mathcal{S}(X)$. To resolve this issue, we use the ideas from Dense Subset Sum~\cite{BW21, GM91}.

Consider the ideal case where no integer $d \geq 2$ can divide lots of the input integers. 
In this case, we can select a small subset $R$ of the input integers such that $\mathcal{S}(R)$ contains all the possible remainders modulo $\Delta$ for all $\Delta \leq \sqrt{t/w}$. Then we apply Algorithm~\ref{alg:sketch} to the remaining integers $X \setminus R$. Suppose that it ends with the dense case. Then we can show that $\mathcal{S}(X \setminus R)$ contains a long arithmetic progression with common difference $\Delta \leq \sqrt{t/w}$. Combining $\mathcal{S}(R)$ and $\mathcal{S}(X\setminus R)$, we have that $\mathcal{S}(X)$ contains an arithmetic progression with common difference $1$. In other words, $\mathcal{S}(X)$ contains every integer in a certain interval. Using standard techniques, we can also ensure that $t$ belongs to this interval, which implies an exact answer that $t \in \mathcal{S}(X)$.

In the general case, it may not be possible for $\mathcal{S}(X)$ to contain all the integers in some long interval. For example, if all the integers in $X$ are multiples of some other integer $d$, then $\mathcal{S}(X)$ can contain only integers that are multiples of $d$. In this case, we can compute a divisor $d$ such that 
    \begin{enumerate}[label={(\roman*)}]
        \item most integers in $X$ are multiples of $d$. Let $X'$ be the set of these integers.

        \item $X'/d$ falls into the ideal case. ($X'/d$ represents the set obtained by dividing each integer in $X$ by $d$.)
    \end{enumerate}
Then we know that $\mathcal{S}(X'/d)$ contains every integer in some long interval, or equivalently, $\mathcal{S}(X')$ contains all the multiples of $d$ in some long interval. Let $D$ be the set of these multiples of $d$. Then we can show that to determine whether $t \in \mathcal{S}(X)$, it suffices to determine whether $t \in D + \mathcal{S}(X\setminus X')$, and the latter can be done efficiently since $X \setminus X'$ has small size.

\paragraph{Combining the above two techniques.}
When combining the above two techniques, the dense-or-sparse framework (Algorithm~\ref{alg:sketch}) is applied to $X' \setminus R$ rather than $X$. Note that the contribution of $X' \setminus R$ to $t$ is unknown at the time we apply the framework. This causes trouble because when applying the random permutation technique, we need to know the expected contribution of $S^j_i$ to $t$. Without knowing the expectation, we cannot locate the interval in which, the contribution of $S^j_i$ to $t$ lies with high probability. We resolve this issue as follows. Let $G = X \setminus X'$. Note that the contribution of the contribution of $X' \setminus R$ to $t$ is at least $t - \sigma(G) - \sigma(R)$ and at most $t$. ($\sigma(G)$ and $\sigma(R)$ represents the sum of all integers in $G$ and $R$,  respectively.) Therefore, even without knowing the exact contribution, we can still locate an interval of length $2\sqrt{wt} + \sigma(R) + \sigma(G)$ in which the contribution of $S^j_i$ to $t$ lies with high probability. As long as $\sigma(R)$ and $\sigma(G)$ is bounded by $\widetilde{O}(\sqrt{wt})$, the random permutation still works. We show an interesting coincidence that the probabilistic and number theoretic arguments meet at $\sigma(G)=\widetilde{\Theta}(\sqrt{wt})$ and $\sigma(R)=\widetilde{\Theta}(\sqrt{wt})$.

The above discussion is based on several simplifications and requires that $t=\Theta(\sum_i x_i)$. When $t$ is significantly smaller than $\sum_i x_i$, we need to further incorporate the color-coding technique from Bringmann~\cite{Bri17}. Moreover, the random permutation technique has to be adjusted to work well with color-coding.

\subsection{Further Related Works}
A more general problem, Bounded Subset Sum, has also been studied in the literature. In Bounded Subset Sum, each input integer $x_i$ can be used up to $u_i$ times. Polak, Rohwedder and W{\k{e}}grzycki \cite{PRW21} showed an $\widetilde{O}(n+w^{5/3})$-time algorithm, which was later improved to an $\widetilde{O}(\min\{n+w^{3/2},nw\})$-time algorithm by Chen, Lian, Mao and Zhang~\cite{CLMZ24aSODA}. Moreover, our algorithm extends naturally to Bounded Subset Sum and runs in time $\widetilde{O}(\sum u_i+\sqrt{wt})$, as the color-coding step requires $O(\sum u_i)$ time.

A closely related problem is Unbounded Subset Sum, where each input integer can be used arbitrarily many times. Bringmann~\cite{Bri17} showed an $\widetilde{O}(n+t)$-time algorithm, followed by an $\widetilde{O}(n+w)$-time algorithm by Jansen and Rohwedder~\cite{JR23}. It is worth mentioning that Unbounded Subset Sum also admits an $w_{\min}^2/2^{\Omega(\log w_{\min})^{1/2}}$-time algorithm by Klein \cite{Kle22}, where $w_{\min}$ refers to the smallest integer in the input.

Another variant of Subset Sum is the Modular Subset Sum problem, where all additions are modulo $m$, for some given integer $m$. The dynamic programming algorithm of Bellman~\cite{Bel57} applies to it in $O(nm)$ time. In recent years, Koiliaris and Xu~\cite{KX18} gave a $\widetilde{O}(\min(\sqrt{n}m,m^{5/4}))$ time algorithm for this problem, which was later improved to $O(m \log^7 m)$ randomized time~\cite{ABJ+19}. Axiotis et al.~\cite{ABB+} simplified their algorithm and presented a randomized algorithm running in time $O(m \log^2 m)$ and a deterministic algorithm running in time $O(m \mathrm{polylog} m)$. At the same time, Cardinal and Iacono~\cite{CI21} presented a randomized algorithm running in time $O(m\log m)$. The currently best-known deterministic algorithm is due to Pot\k{e}pa~\cite{Pot21}, which runs in time $O(m\log m\cdot \alpha(m))$, where $\alpha(m)$ is the inverse Ackerman function.

In addition to pseudo-polynomial time algorithms for Subset Sum, there is also a line of research on approximation algorithms~\cite{IK75, Kar75, KPS97, BN21b, MWW19, DJM23, WC22, CLMZ24cSTOCPartition}.

 Pseudo-polynomial time algorithms for Knapsack have also been extensively studied in the literature~\cite{Bel57, Pis99, Tam09, BHSS18, AT19, BHSS18, PRW21, BC22, BC23, CLMZ24aSODA, Jin24, Bri24, HX24}.

Our algorithm leverages sparse Fast Fourier Transform, see, e.g.,~\cite{CH02, AR15, CL15, Nak20, GGC20, BFN21, BN21a, BFN22}. 
We also use results from additive combinatorics~\cite{Alo87, Sar89, Sar94, SV05, SV06}. It is worth mentioning that a long line of research works used
these additive combinatorics results to design algorithms for dense cases of Subset Sum~\cite{Cha99, CFG89, GM91, BW21}. Notably, most of these works assume that the input is a set, and only
recently has it been generalized to multi-sets~\cite{BW21} and has been adopted to obtain faster algorithms for Subset Sum and Knapsack \cite{ MWW19, BN20, BW21, WC22, DJM23, CLMZ24aSODA, Jin24, Bri24}.

\subsection{Paper Outline}
In Section~\ref{sec:pre}, we introduce all the necessary terminology and known facts. In Section~\ref{sec:keylem}, we present the framework for proving our main theorem, where a lemma plays an important role. This key lemma will be proved in the next two sections. Section~\ref{sec:dense} tackles the dense part of the lemma using additive combinatorics tools, while Section~\ref{sec:sparse} tackles the sparse part using color-coding, random permutation and sparse convolution. 

\section{Preliminaries}\label{sec:pre}
Throughout the paper, we distinguish multi-sets from sets. Only when we explicitly use the term multi-sets do we allow duplicate elements. The subset of a multi-set is always a multi-set unless otherwise stated. All logarithms in this paper are base 2. When we say ``with high probability'', we mean with probability at least $1 - (n+t)^{-\Omega(1)}$, where $n$ is the size of $X$ and $t$ is the target.

\subsection{Notation}
Let $Z$ be a multi-set of non-negative integers. We use $\max(Z)$ and $\min(Z)$ to denote the maximum and minimum elements of $Z$.  The diameter of $Z$ is defined as
\(
    \dm(Z) = \max(Z) - \min(Z)+1.
\) 
We write $\sum_{z \in Z}z$ as $\sigma(Z)$. For emptyset, we define $\max(\emptyset) = \min(\emptyset) = \sigma(\emptyset) = 0$. We define $\mathcal{S}(Z)$ to be the set of all subset sums of $Z$. That is, 
\(
    \mathcal{S}(Z) = \{\sigma(Y) : Y \subseteq Z\}.
\)

Given any two integers $a$ and $b$, we use $[a,b]$ to denote the set of integers between $a$ and $b$, i.e., $[a,b] = \{z \in\mathbb{Z}: a\leq z\leq b\}$.

\subsection{Sumsets and Convolution}
Let $A$ and $B$ be two sets of integers. Their sumset $A+B$ is defined as follows.
    \[
        A + B = \left\{a + b : a\in A, b \in B\right\}
    \]
It is well-known that computing sumsets is equivalent to computing the convolution of two vectors, which can be solved by Fast Fourier Transformation.
\begin{lemma}\label{lem:fft}
    Let $A$ and $B$ be two sets of integers. We can compute their sumset $A + B$ in $O(u\log u)$ time where $u = \max\{\dm(A), \dm(B), 2\}$.
\end{lemma}

When $|A+B|$ is small, sparse convolution algorithms, whose running time depends on the output size, are much faster than the classic algorithm. We use the best deterministic algorithm, which is due to Bringmann, Fischer, and Nakos~\cite{BFN22}. This is to ease the analysis. One may replace it with a faster randomized algorithm (e.g.,\cite{CH02,Nak20,GGC20,BFN21}) to reduce the logarithmic factors in the running time.

\begin{lemma}[\normalfont \cite{BFN22}]\label{lem:sparse-fft}
    Let $A$ and $B$ be two sets of integers. We can compute their sumset $A + B$ in $O(|A + B| \mathrm{polylog}\,u)$ time where $u = \max\{\dm(A), \dm(B), 2\}$.
\end{lemma}

\subsection{Subset Sum} 

Let $(X,t)$ be a Subset Sum instance. We always assume that $t \leq \sigma(X)/2$.  This is without loss of generality: when $t > \sigma(X)/2$, we can let $t' = \sigma(X) - t$ and solve the instance $(X, t')$. Let $w$ be the maximum element of $X$. 
We assume that $n, t, w$ are greater than any constant appearing in this paper, since otherwise the instance can be solved in $O(n)$ time by known algorithms.  We further assume that $t \geq 10\sqrt{wt}\log w$, since otherwise, the instance can be solved in time $\widetilde{O}(n + \sqrt{wt})$ by the $\widetilde{O}(n + t)$-time algorithm~\cite{Bri17}.

\section{Framework for Proving the Main Theorem}\label{sec:keylem}

We first define some notation that is only used in this section. Let $Z$ be a multi-set and $d$ be an integer. We define the multi-set $dZ = \{dz : z\in Z\}$ and the set $Z \bmod d = \{z \bmod d : z \in Z\}$. We define $Z(d)$ to be the multi-set of all integers in $Z$ that are divisible by $d$ and $\overline{Z(d)} = Z\setminus Z(d)$ to be the multi-set of all integers in $Z$ that are not divisible by $d$. If every integer in $Z$ is divisible by $d$, we define the multi-set $Z/d = \{z/d : z\in Z\}$. 

As in many previous works, we reduce Subset Sum to a problem of computing a subset of $\mathcal{S}(X)$. (Recall that $\mathcal{S}(X)$ is the set of all subset sums of $X$.) More precisely, we compute a set $S \subseteq \mathcal{S}(X)$ of size $\widetilde{O}(\sqrt{wt})$ such that $S$ contains $t$ with high probability if $t \in \mathcal{S}(X)$. This allows us to solve the instance $(X, t)$ in $\widetilde{O}(\sqrt{wt})$ time by checking whether $t \in S$.

Our approach builds upon the framework proposed for Dense Subset Sum ~\cite{GM91,BW21}. 
We roughly explain the idea. If $X$ does not contain an almost divisor, that is, an integer $d$ that divides almost all elements of $X$, then we can find a small subset $R$ such that for any possible common difference $d'$, $\mathcal{S}(R)\bmod d'$ covers all possible values. In this case, we show that if $\mathcal{S}(X\setminus R)$ contains an arithmetic progression, then $\mathcal{S}(X)$ must contain an interval that includes $t$.
On the other hand, if $X$ does have an almost divisor $d$, we can extract a small subset $G$ such that all elements in $X\setminus G$ are divisible by $d$ and $(X \setminus G)/d$ does not have any almost divisor. This brings us back to the previous case. Finally, we reconstruct the full set via convolution $\mathcal{S}(X)=\mathcal{S}(G)+d\mathcal{S}((X\backslash G)/d)$.

We begin by presenting three lemmas that form the basis for the proof of our main theorem. The first and third lemmas are slightly modified from Dense Subset Sum~\cite{BW21}, which are proved in the next two subsections, respectively. The second lemma serves as our main technical contribution, where the remaining part of the paper is devoted to its proof.

The first lemma formally defines the partition of $X$ into $G$, $R$, and $D=X\setminus(G\cup R)$.
\begin{restatable}{lemma}{lemdensesubsetsum}
    \label{lem:dense-subset-sum}
    Given a multi-set $X$ of $n$ positive integers from $[1, w]$ and a positive integer $t$, in time $\widetilde{O}(n + \sqrt{w})$, we can compute an integer $d \geq 1$ and a partition $G\cup R \cup D$ of $X$ such that the following holds.
    \begin{enumerate}[label={\normalfont (\roman*)}] 
        \item $\sigma(G) \leq \sqrt{wt} \log w$ and $\sigma(R) \leq 4\sqrt{wt}\log w$.
        \item Every integer in $R \cup D$ is divisible by $d$.
        \item $\mathcal{S}({R/d}) \bmod b = [0, b-1]$ for any $b \in [1, \sqrt{t/w}]$. That is, for any $b \in [1, \sqrt{t/w}]$, the multi-set $R/d$ can generate all the remainders modulo $b$.
    \end{enumerate}
\end{restatable}

Both $\sigma(G)$ and $\sigma(R)$ are bounded by $\widetilde{O}(\sqrt{wt})$. We can easily deal with $\mathcal{S}(G)$ and $\mathcal{S}(R)$ using known algorithms. As for $D$, if some subset of $X$ sums to $t$, we know that $D$ contributes at most $t$ and at least $t - \sigma(G)-\sigma(R) \geq t - 5\sqrt{wt}\log w$. Therefore, it suffices to compute $\mathcal{S}(D) \cap [t - 5\sqrt{wt}\log w, t]$.  The following lemma is the key to our result. It states that, in $\widetilde{O}(n + \sqrt{wt})$ time, we can either deal with $\mathcal{S}(D) \cap [t - 5\sqrt{wt}\log w, t]$ or show that some small subset of $D$ can generate a long arithmetic progression. 

\begin{restatable}{lemma}{lemkey}\label{lem:key}
    Let $D$ be a multi-set of at most $n$ positive integers from $[1, w]$. Let $t$ be a positive integer. In $\widetilde{O}(n + \sqrt{wt})$ time, we can obtain one of the following results.
    \begin{itemize}
        \item{Sparse case:} we can obtain a set $S \subseteq \mathcal{S}(D) \cap [t - 5\sqrt{wt}\log w, t]$ containing any $s \in \mathcal{S}(D) \cap [t - 5\sqrt{wt}\log w, t]$ with probability $1 - (n + t)^{-\Omega(1)}$, or
    
        \item {Dense case:} we can show that the following holds for some subset $P \subseteq D$.
        \begin{enumerate}[label={\normalfont (\roman*)}]
            \item The set $\mathcal{S}(P)$ contains an arithmetic progression $(a_1, \ldots, a_k)$ with $k \geq 5\sqrt{wt}\log w$ and $a_k \leq t/2$.

            \item $\sigma(D) - \sigma(P) \geq t$.
        \end{enumerate}
    \end{itemize}
\end{restatable}

If Lemma~\ref{lem:key} ends with the sparse case, we can solve the Subset Sum instance easily via convolution. If it ends with the dense case, however, we need extra tools from Dense Subset Sum~\cite{BW21} to show that $\mathcal{S}(R \cup D)$ contains all the multiples of $d$ within $[t - \sqrt{wt}\log w, t]$. 

\begin{restatable}{lemma}{lemaugmentap}
\label{lem:augment-ap}
    Let $R$, $P$, and $Q$ be multi-sets of integers from $[1, w]$. Suppose that $\mathcal{S}(P)$ contains an arithmetic progression $(a_1, \ldots, a_k)$ with common difference $\Delta$ and that $\mathcal{S}(R) \bmod \Delta = [0, \Delta - 1]$. If $k \geq \sigma(R) + w + 1$, then $\mathcal{S}(R \cup P \cup Q)$ contains every integer in $[\sigma(R) + a_k, \sigma(Q)]$. In other words,
    \[
        [\sigma(R) + a_k, \sigma(Q)] \subseteq \mathcal{S}(R \cup P \cup Q).
    \]
\end{restatable}

Now we are ready to prove our main theorem.
\thmmain*
\begin{proof}
    It suffices to show that, in $\widetilde{O}(n + \sqrt{wt})$ time, we can compute a set $S \subseteq \mathcal{S}(X)$ of size $\widetilde{O}(\sqrt{wt})$ such that $S$ contains $t$ with probability $1 - (n + t)^{-\Omega(1)}$ if $t \in \mathcal{S}(X)$. Recall that to simplify the notation, the term ``with high probability'' means ``with probability at least $1 - (n + t)^{-\Omega(1)}$'' throughout this proof.

    We first preprocess $X$ via Lemma~\ref{lem:dense-subset-sum}, and get an integer $d$ and a partition $G \cup R \cup D$ of $X$. Since $\sigma(G)$ is bounded by $\widetilde{O}(\sqrt{wt})$, in $\widetilde{O}(n + \sqrt{wt})$ time, we can obtain a set $S_G \subseteq \mathcal{S}(G)$ containing any $s \in \mathcal{S}(G)$ with high probability. This can be done by Bringmann's $\widetilde{O}(n+t)$-time algorithm~\cite{Bri17}.  Similarly, in $\widetilde{O}(n + \sqrt{wt})$ time, we can compute a set $S_R \subseteq \mathcal{S}(R)$ containing any $s \in \mathcal{S}(R)$ with high probability. Then we process $D$ via Lemma~\ref{lem:key}. There are two cases.

    Suppose the Lemma~\ref{lem:key} ends with the sparse case. That is, we obtain a set $S_D \subseteq \mathcal{S}(D) \cap [t - 5\sqrt{wt}\log w, t]$ containing any $s \in \mathcal{S}(D) \cap [t - 5\sqrt{wt}\log w, t]$ with high probability. Note that the diameters of $S_G$, $S_R$, and $S_D$ are all bounded by $\widetilde{O}(\sqrt{wt})$; we can compute $S = S_G + S_R + S_D$ via Lemma~\ref{lem:fft} in $\widetilde{O}(\sqrt{wt})$ time. Next, we show that $S$ contains $t$ with high probability if $t \in \mathcal{S}(X)$. Let $Z \subseteq X$ be some subset with $\sigma(Z) = t$. It is easy to see that $S_G$ and $S_R$ contain $\sigma(Z \cap G)$ and $\sigma(Z \cap R)$ with high probability, respectively. Consider $\sigma(Z \cap D)$. We have
    \[
        t \geq \sigma(Z\cap D) \geq \sigma(Z) - \sigma(G) - \sigma(R) \geq t - 5\sqrt{wt}\log w.
    \]
    Therefore, $S_D$ contains $\sigma(Z \cap D)$ with high probability. Then $S = S_G + S_R + S_D$ contains $\sigma(Z) = t$ with high probability.

    Suppose that Lemma~\ref{lem:key} ends with the dense case.  We claim that in this case, $\mathcal{S}(R \cup D)$ contains all the multiples of $d$ within $[t - \sqrt{wt}\log w, t]$. Assume that the claim holds. Since every integer in $R\cup D$ is divisible by $d$ (Lemma~\ref{lem:dense-subset-sum}(ii)), the claim implies that 
    \[
        \mathcal{S}(R \cup D)\cap [t - \sqrt{wt}\log w, t] = [t - \sqrt{wt}\log w, t] \cap d\mathbb{N}.
    \]
    Let $S_{RD} = [t - \sqrt{wt}\log w, t] \cap d\mathbb{N}$. Let $S = S_G + S_{RD}$. The set $S$ can be computed in $\widetilde{O}(\sqrt{wt})$ time since the diameters of $S_G$ and $S_{RD}$ are bounded by $\widetilde{O}(\sqrt{wt})$. Suppose that some subset $Z \subseteq X$ has $\sigma(Z) = t$. Then $S_G$ contains $\sigma(Z \cap G)$ with high probability. Moreover, $S_{RD}$ must contain $\sigma(Z \cap (R \cup D))$ since $t \geq \sigma(Z \cap (R \cup D)) \geq \sigma(Z) - \sigma(G) \geq t - \sqrt{wt}\log w$. Then we have that $S$ contains $t = \sigma(Z)$ with high probability.

    It remains to prove the claim that $\mathcal{S}(R \cup D)$ contains all the multiples of $d$ within $[t - \sqrt{wt}\log w, t]$ in the dense case. Recall that every integer in $R\cup D$ is divisible by $d$. Let $R' = R/d$ and $D' = D/d$. It suffices to show that $[(t - \sqrt{wt}\log w)/d, t/d] \subseteq \mathcal{S}(R' \cup D')$. Lemma~\ref{lem:key}(i) guarantees that $D'$ has a subset $P'$ such that $\mathcal{S}(P')$ contains an arithmetic progression $(a_1, \ldots, a_k)$ with $k \geq 5\sqrt{wt}\log w$ and $a_k \leq t/(2d)$. The common difference of the arithmetic progression is $\Delta < a_k/k < \sqrt{t/w}$. Lemma~\ref{lem:dense-subset-sum}(iii) guarantees that $\mathcal{S}(R') \bmod \Delta = [0, \Delta-1]$. One can verify that $R'$, $P'$, and $D'\setminus P'$ satisfy all the conditions of Lemma~\ref{lem:augment-ap}. Therefore, we have 
    \[
        [\sigma(R') + a_k, \sigma(D'\setminus P')] \subseteq \mathcal{S}(R' \cup D').
    \]
    Note that 
    \[
        \sigma(R') + a_k \leq \frac{4\sqrt{wt}\log w}{d} + \frac{t}{2d} \leq \frac{t - \sqrt{wt}\log w}{d}.
    \]
    The last inequality is due to our assumption that $t \geq 10\sqrt{wt}\log w$. Lemma~\ref{lem:key}(ii) guarantees that
    \[
        \sigma(D' \setminus P') = \frac{\sigma(D) - \sigma(P)}{d} \geq \frac{t}{d}.
    \]
    The above three inequalities imply that
    \[
        \left[\frac{t - \sqrt{wt}\log w}{d}, \frac{t}{d}\right] \subseteq \mathcal{S}(R' \cup D').
    \]
    This completes the proof.
\end{proof}

Section~\ref{sec:dense} and Section~\ref{sec:sparse} are devoted to proving the key lemma (Lemma~\ref{lem:key}). Recall that we assume $10\sqrt{wt}\log w\leq t \leq \sigma(X)/2$.
Then we have
\[
    \sigma(D) \geq \sigma(X) - \sigma(G) - \sigma(R) \geq \sigma(X) - 5\sqrt{wt}\log w \geq \frac{3t}{2}.
\]
In the rest of the paper, we always assume that $\sigma(D) \geq \frac{3t}{2}$. Before moving further, we complete this section by presenting the proofs for Lemmas ~\ref{lem:dense-subset-sum} and ~\ref{lem:augment-ap}.

\subsection{Proof of Lemma~\ref{lem:dense-subset-sum}}

We first formally define almost divisors and show that such the subset $G$ can be efficiently computed.

\begin{definition}
    Let $X$ be a multi-set of positive integers. We say an integer $d > 1$ is an $\alpha$-almost divisor of $X$ if $|\overline{X(d)}| \leq \alpha$. 
\end{definition}

\begin{lemma}[implied by {\cite[Theorem 4.1]{BW21}}]
\label{lem:part-almost-div}
    Given $\alpha > 0$ and a multi-set $X$ of $n$ positive integers from $[1, w]$, in time $\widetilde{O}(n + \sqrt{w})$ time, we can compute a divisor $d \geq 1$ such that the following holds.
    \begin{enumerate}[label={\normalfont (\roman*)}]
        \item $X(d)/d$ has no $\alpha$-almost divisor.

        \item $|\overline{X(d)}| \leq \alpha \log w$.
    \end{enumerate}
\end{lemma}
\begin{proof}
We can check whether $X$ has an $\alpha$-almost divisor in $\widetilde{O}(n + \sqrt{w})$ time via the following lemma from \cite{BW21}. 

\begin{lemma}[{\cite[Theorem 4.12]{BW21}}]\label{lem:find-almost-divisor}
    Given $\alpha > 0$ and a multi-set $X$ of $n$ positive integers from $[1, w]$, in $\widetilde{O}(n + \sqrt{w})$ time, we can decide whether $X$ has an $\alpha$-almost divisor, and compute an $\alpha$-almost divisor if it exists.
\end{lemma}

If $X$ has no almost divisor, we let $d = 1$ and stop. Suppose that $X$ has an almost divisor. Starting with $X_1 = X$, we iteratively find and remove almost divisors. That is, if $X_{i}$ has an almost divisor $d_i$, we continue with $X_{i+1} := X_{i}(d_i)/d_i$. We stop when the multi-set $X_{k}$ has no almost divisor, and let $d:= d_1\cdots d_{k-1}$.

It is easy to see that $X(d)/d = X_k$, so it has no almost divisor. Since $d = d_1\cdots d_{k-1} \leq w$ and $d_i\geq 2$ for all $i$, the number of iterations is bounded by $\log w$. Then 
\[
|\overline{X(d)}| = \sum_{i=1}^k |\overline{X_{i-1}(d_{i})}|\leq \sum_{i=1}^k\alpha\leq \alpha\log w.
\]
Since there are at most $\log w$ iterations and each costs $\widetilde{O}(n+\sqrt{w})$ time, the total running time is $\widetilde{O}(n+\sqrt{w})$.
\end{proof}

$\overline{X(d)}$ is in fact the set $G$. Now we show that if $X$ does not have almost divisors, we can find a small subset $R$ that for any possible common difference, $R$ can generate all the remainders modulo it.

\begin{lemma}\label{lem:part-gen-rem}
    Given $\alpha > 0$ and a multi-set $X$ of $n$ positive integers from $[1, w]$, if $X$ has no $\alpha$-almost divisor, in $\widetilde{O}(n + \sqrt{w})$ time, we extract a set $R \subseteq X$ such that the following holds.
    \begin{enumerate}[label={\normalfont (\roman*)}]
        \item $|R| \leq 4\alpha\log w$.

        \item $\mathcal{S}_R \bmod d = [0, d-1]$ for any $1 < d \leq \alpha$.
    \end{enumerate}
\end{lemma}
\begin{proof}
We first show that if $X$ has no $\alpha$-almost divisor, in $\widetilde{O}(n + \sqrt{w})$ time, we can extract a subset $R \subseteq X$ such that $|R| \leq 4\alpha\log w$ and for any $1 < d \leq \alpha$, the multi-set $R$ contains at least $d$ integers not divisible by $d$. That is, $|\overline{R(d)}| \geq d$.

Pick an arbitrary subset $R'\subseteq X$ of size $2\alpha$. Let $P$ be the set of primes $p$ with $p\leq \alpha$ and $|\overline{R'(p)}|\leq \alpha$. We have that
\begin{claim}[{\cite[Claim 4.21]{BW21}}]
     $|P|\leq 2\log w$.
\end{claim}
We can compute the prime factorization of all elements in $R'$ with the following lemma. Therefore, $P$ can be carried out in $\widetilde{O}(\alpha+\sqrt{w})$. 
\begin{lemma}[{\cite[Theorem 3.8]{BW21}}]\label{lem:prime-factorize}
    The prime factorization of $n$ given numbers in $[1,w]$ can be computed in $\widetilde{O}(n + \sqrt{w})$ time.
\end{lemma}
For any $p\in P$, let $R_p \subseteq \overline{X(p)}$ be an arbitrary subset of size $\alpha$. Since $X$ has no almost divisor, $R_p$ always exists and it can be found in $O(n)$ time.

Let $R=R'\cup \bigcup_{p\in P}R_p$. For the first property, note that 
\[
|R|\leq |R'|+\sum_{p\in P} |R_p|\leq 2\alpha+|P|\cdot\alpha\leq  4\alpha\log w.
\]
For any integer $1<d\leq \alpha$, we have $|\overline{R(d)}|\geq |\overline{R(p)}| \geq 2$, where $p$ is an arbitrary prime factor of $d$. Now we can use the following lemma to conclude the proof. 

\begin{lemma}[{\cite[Theorem 4.22]{BW21}}]\label{lem:all-rem}
    Given $\alpha > 0$ and a multi-set $R$, if $|\overline{R(d)}| \geq d$ for any $1 < d \leq \alpha$, then $\mathcal{S}_R \bmod d = [0, d-1]$ for any $1 < d \leq \alpha$.\qedhere
\end{lemma}
\end{proof}

Now we can prove Lemma~\ref{lem:dense-subset-sum}.

\lemdensesubsetsum*

\begin{proof}
    Let $\alpha = \sqrt{t/w}$. We first compute $d$ by using Lemma~\ref{lem:part-almost-div}.  Let $G = \overline{X(d)}$ and $X' = X(d)/d$.  Since $X'$ has no $\alpha$-almost divisor, we can extract a set $R'$ from $X'$ via Lemma~\ref{lem:part-gen-rem}. Let $R = dR'$ and $D = d(X'\setminus R')$. The total time cost is $\widetilde{O}(n + \sqrt{w})$. All the stated properties can be easily verified. 
\end{proof}

\subsection{Proof of Lemma~\ref{lem:augment-ap}}

\lemaugmentap*

\begin{proof}
    Let $y$ be an arbitrary integer in $[\sigma(R) + a_k, \sigma(Q)]$. We shall prove $y \in \mathcal{S}(R\cup P\cup Q)$ by showing $y = \sigma(R') + \sigma(P') + \sigma(Q')$ for some $R' \subseteq R$, $P' \subseteq P$, and $Q' \subseteq Q$. 

    We determine $Q'$ first. Since $0\leq y - a_{k}\leq \sigma(Q)$, there is some $Q' \subseteq Q$ such that $y - a_{k} \leq \sigma(Q') \leq y - a_k + w$. Equivalently,
   \[
      a_k - w \leq y - \sigma(Q') \leq a_k.
   \]
   Since $a_{k} - w = a_1 + (k-1)\Delta - w \geq a_1 + \sigma(R)$, we have $a_1 + \sigma(R)\leq y - \sigma(Q')\leq a_k$.

   Next, we determine $R'$. Let $R'$ be a subset of $R$ with $\sigma(R') \equiv y - \sigma(Q') - a_1 \pmod \Delta$.  $R'$ must exist since $\mathcal{S}_R \bmod \Delta = [0, \Delta-1]$.  Now consider $y - \sigma(Q') - \sigma(R')$. We have $y - \sigma(Q') - \sigma(R') \equiv a_1 \pmod \Delta$ and
   \[
       a_1 \leq a_1 + \sigma(R) - \sigma(R')\leq y - \sigma(Q') - \sigma(R') \leq a_{k}.
   \]
   Therefore, $y - \sigma(D') - \sigma(R')$ is a term in the arithmetic progression $\{a_1, \ldots, a_k\}$; there must be a set $P' \subseteq P$ with $\sigma(P') = y - \sigma(Q') - \sigma(R')$.
\end{proof}

\section{The Dense Case}\label{sec:dense}
In this section, we use additive combinatorics tools to characterize the condition under which $D$ will have a subset $P$ that satisfies Lemma~\ref{lem:key}(i) and (ii).  This section is completely structural, and hence, does not depend on the algorithm. 

The following lemma is the main tool we will use in this section. It simplifies and generalizes a similar result in~\cite{CLMZ24cSTOCPartition}. Basically, given many sets with a large total size, we can select some sets so that their sumset has an arithmetic progression. In addition, if we require the total size to be around $\rho$ times larger than the threshold, then we can select some sets so that the sumset of any $1/\rho$ fraction of these selected sets has an arithmetic progression. In our usage, we always choose those with small weights (small $f$ value).

\begin{restatable}{lemma}{lemapwithsmallelement}
\label{lem:ap-with-small-element}
    There exists some constant $C_{\ap}$ such that the following is true. Let $A_1, \ldots, A_{\ell}$ be non-empty subsets of integers. Let $f: \{A_i\}_{i=1}^\ell \to \mathbb{N}$ be a weight function.  Let $u$ be the maximum diameter of all the $A_i$'s. If $\sum_{i=1}^\ell |A_i| \geq \ell + 4C_{\ap}\rho u'\log u'$ for some $u' \geq u$ and some $\rho \geq 1$, then we can select a collection of sets $A_i$ such that the sumset of the selected sets contains an arithmetic progression of length at least $u'$. Let $I$ be the set of the indices of the selected sets $A_i$. We can also guarantee that
    \[
        \sum_{i\in I}f(A_i) \leq \frac{1}{\rho}\sum_{i=1}^{\ell}f(A_i).
    \]
\end{restatable}

To prove it, we first show an additive combinatorics result from Szemer{\'e}di and Vu~\cite{SV05}.

\begin{theorem}[Corollary 5.2~\cite{SV05}]\label{thm:ap}
    For any fixed integer $d$, there are positive constants $c_1$ and $c_2$ depending on $d$ such that the following holds.  Let $A_1, \ldots, A_{\ell}$ be subsets of $[1,u]$ of size $k$.  If $\ell^d k \geq c_1u$, then $A_1 + \cdots + A_\ell$ contains an arithmetic progression of length at least $c_2\ell k^{1/d}$.
\end{theorem}
Although not explicit in the statement, the above theorem actually assumes that $k \geq 2$. Therefore, we have the following corollary.

\begin{corollary}\label{coro:ap}
    There exists a sufficiently large constant $C_{\ap}$ such that the following holds.  Let $A_1, \ldots, A_{\ell}$ be subsets of $[1,u]$ of size at least $k$. If $k \geq 2$ and $\ell k \geq cu'$ for some $u' \geq u$, then $A_1+ \cdots + A_\ell$ contains an arithmetic progression of length at least $u'$.
\end{corollary}
\begin{proof}
    Let $c_1$ and $c_2$ are two constants for $d=1$ in Theorem~\ref{thm:ap}.  Assume that $C_{\ap} \geq c_1$ and that $C_{\ap}c_2 \geq 1$ since $C_{\ap}$ is sufficiently large.  Since $\ell k \geq C_{\ap}u' \geq c_1u$, by Theorem~\ref{thm:ap}, $A_1 + \cdots + A_\ell$ contains an arithmetic progression of length at least $c_2\ell k \geq c_2C_{\ap} u' \geq  u'$.
\end{proof}

The requirement that each subset be large can be relaxed to a condition on the total size of all subsets, which is often more convenient to work with. Furthermore, by increasing the threshold for density, it suffices to use only a smaller subset to construct an arithmetic progression.

\begin{lemma}
\label{lem:ap-with-small-element-0}
    There exists some constant $C_{\ap}$ such that the following is true. Let $A_1, \ldots, A_{\ell}$ be non-empty subsets of $[1,u]$ for some positive integer $u$. If $\sum_{i=1}^\ell |A_i| \geq \ell + 4C_{\ap}\rho u'\log u'$ for some $u' \geq u$ and some $\rho \geq 1$, then we can select a collection $\{A_i\}_{i\in I}$ such that $\sum_{i\in I} A_i$ contains an arithmetic progression of length at least $u'$. If we are also given a function $f$ that maps each $A_i$ to a non-negative integer, we can also guarantee that 
    \[
        \sum_{i\in I}f(A_i) \leq \frac{1}{\rho}\sum_{i=1}^{\ell}f(A_i).
    \]
\end{lemma}
\begin{proof}
    When $u=1$, the proof is trivial.  Assume that $u \geq 2$.  Let $c$ denote $C_{\ap}$.  We claim that for some $k \in [2, u]$, at least $\frac{2c\rho u'}{k}$ sets from $\{A_1, \ldots, A_\ell\}$ have size at least $k$.  Assume that the claim holds. We select $\lceil\frac{cu'}{k}\rceil$ such $A_i$'s greedily in the sense that we always prefer those with small $f(A_i)$.  Let $I$ be the index of the selected $A_i$'s. Clearly, $\{A_i\}_{i\in I}$ satisfies the condition of Corollary~\ref{coro:ap}, and hence, $\sum_{i\in I} A_i$ contains an arithmetic progression of length at least $u'$.  Since we select $\lceil\frac{cu'}{k}\rceil$ out of $\frac{2c\rho u'}{k}$ sets $A_i$'s and we prefer those with small $f(A_i)$, 

       \[
        \frac{\sum_{i \in I} f(A_i)}{\sum_{i=1}^\ell f(A_i)} \leq \frac{\lceil\frac{cu'}{k}\rceil}{\frac{2c\rho u'}{k}}\leq \frac{\frac{2cu'}{k}}{\frac{2c\rho u'}{k}} \leq \frac{1}{\rho}.
    \]
    The second inequality is due to that $\frac{cu'}{k} \geq \frac{cu}{k} \geq c \geq 1$.

    We prove the claim by contradiction.  Let $\ell_k$ be the the number of sets from $\{A_1, \ldots, A_\ell\}$ that have size at least $k$. Suppose that $\ell_k < \frac{2c\rho u'}{k}$ for any $k \in [2, u]$.
    Then 
    \begin{align*}
        \sum_{i=1}^{\ell}|A_i| &= \sum_{k=1}^{u} k(\ell_k - \ell_{k+1}) \\&=\sum_{k=1}^{u}\ell_k \\&<\ell_1 + \sum_{k=2}^{u} \frac{2c \rho u'}{k} \\&\leq \ell + 2c\rho u' (1 + \log u) \\&\leq \ell + 4c\rho u'\log u'.
    \end{align*}
    Contradiction. (The last inequality is due to our assumption that $u \geq 2$.) 
\end{proof}

Now we can generalize it to Lemma~\ref{lem:ap-with-small-element}.
\begin{proof}[Proof of Lemma~\ref{lem:ap-with-small-element}]
    For each $A_i$, we define $A'_i = A_i - \min(A_i) + 1$.  Then we have each $A'_i$ as a subset of $[1, u]$. By Lemma~\ref{lem:ap-with-small-element-0}, we can a select collection $\{A'_i\}_{i\in I}$ such that $\sum_{i\in I} A'_i$ contains an arithmetic progression of length at least $u'$ for some $u'\geq u$. This implies that $\sum_{i\in I} A_i$ contains an arithmetic progression of length at least $u'$. Moreover, if we define $f(A'_i) = f(A_i)$, then Lemma~\ref{lem:ap-with-small-element-0} guarantees that 

    \[
    \sum_{i\in I}f(A_i) \leq \frac{1}{\rho}\sum_{i=1}^{\ell}f(A_i).\qedhere
    \] 

\end{proof}

Below we give a condition under which $D$ is dense, and prove that, if dense, $D$ has a subset $P$ that satisfies Lemma~\ref{lem:key}(i) and (ii).

\begin{definition}\label{def:dense}
    Let $D$ be a multi-set of integers from $[1,w]$. Let $t$ be a positive integer. We say $D$ is \textbf{dense} with respect to $t$ if there exist a positive integer $\rho$, a partition $D_1, \ldots, D_\ell$ of $D$, a set $S_i \subseteq \mathcal{S}(D_i)$ for each $D_i$, and a weight function $f: \{S_i\}_{i=1}^\ell \to \mathbb{N}$ such that the following is true.
    \begin{enumerate}[label={\normalfont (\roman*)}]
        \item $\sum_{i=1}^\ell |S_i| \geq \ell + 4C_{\ap}\rho u' \log u'$ for some $u' \geq \max\{u, $ $5\sqrt{wt}\log w\} $, where $u$ is the maximum diameter of the sets $S_i$ and $C_{\ap}$ is the constant in Lemma~{\normalfont~\ref{lem:ap-with-small-element}}.

        \item $\max(S_i) \leq f(S_i) \leq \sigma(D_i)$ for all $i \in [1, \ell]$.

        \item $\frac{3t}{2} \leq \sum_{i=1}^\ell f(S_i) \leq \frac{\rho t}{2}$.
    \end{enumerate}
\end{definition}

In the above definition, property (i) ensures that there is a long arithmetic progression. As one will see in the proof of the next lemma, properties (ii) and (iii) ensure that the largest term in the arithmetic progression is at most $t/2$, and that the integers that are not used to generate the arithmetic progression have a sum at least $t$.

\begin{lemma}\label{lem:good-ap}
    Let $D$ be a multi-set of integers from $[1,w]$. Let $t$ be a positive integer. If $D$ is dense with respect to $t$, then $D$ has a subset $P$ such that the following hold.
    \begin{enumerate}[label={\normalfont (\roman*)}]
        \item The set $\mathcal{S}(P)$ contains an arithmetic progression $(a_1, \ldots, $ $a_k)$ with $k \geq 5\sqrt{wt}\log w$ and $a_k \leq t/2$.

        \item $\sigma(D) - \sigma(P) \geq t$.
    \end{enumerate} 
\end{lemma}
\begin{proof}
    Let $\rho$, $(D_1,\ldots, D_\ell)$, $(S_1, \ldots, S_\ell)$, and $f$ be as in Definition~\ref{def:dense}. Property (i) of Definition~\ref{def:dense} implies that the sets $S_i$ satisfy the condition of Lemma~\ref{lem:ap-with-small-element}, and hence we can select a collection of sets $S_i$ such that the sumset of the selected sets contains an arithmetic progression $(a_1, \ldots, a_k)$ of length at least $5\sqrt{wt}\log w$. Clearly $k  \geq 5\sqrt{wt}\log w$. Let $I$ be the set of the indices of the selected sets $S_i$.  Definition~\ref{def:dense}(ii)(iii) and Lemma~\ref{lem:ap-with-small-element} guarantees that 
    \begin{equation}\label{eq:f}
            \sum_{i\in I}\max(S_i)\leq \sum_{i \in I}f(S_i) \leq \frac{1}{\rho}\sum_{i=1}^\ell f(S_i) \leq \frac{t}{2}.
    \end{equation}
    Since $a_k \in \sum_{i \in I}S_i$, we have $a_k \leq \sum_{i\in I}\max(S_i) \leq t/2$. 

    Let $P = \bigcup_{i \in I}D_i$. Clearly, $\mathcal{S}(P)$ contains $(a_1, \ldots, a_k)$. Definition~\ref{def:dense}(ii)(iii) and Lemma~\ref{lem:ap-with-small-element} guarantees that
    \begin{align*}  
        \sigma(D) - \sigma(P) &= \sum_{i \notin I} \sigma(D_i) \\
        &\geq \sum_{i \notin I} f(S_i) \\
        &\geq \sum_{i=1}^\ell f(S_i) - \sum_{i \in I} f(S_i) \\
        &\geq \frac{3t}{2} - \frac{t}{2} = t.
    \end{align*}
    The last inequality is due to~\eqref{eq:f} and Definition~\ref{def:dense}(iii).
\end{proof}

\section{The Sparse Case}\label{sec:sparse}
Given Lemma~\ref{lem:good-ap} in the previous section, the key lemma (Lemma~\ref{lem:key}) can be reduced to the following lemma.

\begin{restatable}{lemma}{lemsparse}\label{lem:sparse}
    Let $D$ be a multi-set of at most $n$ positive integers from $[1, w]$. Let $t$ be a positive integer. In $\widetilde{O}(n + \sqrt{wt})$ time, we can obtain one of the following results.
    \begin{itemize}
        \item{Sparse case:} we can obtain a set $S \subseteq \mathcal{S}(D) \cap [t - 5\sqrt{wt}\log w, t]$ containing any $s \in \mathcal{S}(D) \cap [t - 5\sqrt{wt}\log w, t]$ with probability $1 - (n + t)^{-\Omega(1)}$, or

        \item {Dense case:} show that $D$ is dense with respect to $t$. (see Definition{\normalfont~\ref{def:dense}})
    \end{itemize}
\end{restatable}
This section is devoted to proving Lemma~\ref{lem:sparse}. In high-level, our algorithm follows the dense-or-sparse framework (See Algorithm~\ref{alg:sketch} for a Sketch of the framework). It partitions $D$ into several groups and then merges the results of the groups in a tree-like manner using sparse convolution. It stops immediately when it encounters a tree level of large total size, and we will show that in this case, $D$ must be dense with respect to $t$.

The algorithm has three phases. In phase one, we randomly partition $D$ into subsets $D_1, \ldots, D_\ell$ for some $\ell \leq 2n$. In phase two, we compute a set $S_i \subseteq \mathcal{S}(D_i)$ for each $D_i$ by further partitioning and (sparse) convolution in a tree-like manner. In phase three, we merge all sets $S_i$ using sparse convolution in a tree-like manner, and the random permutation technique is used to bound the diameters of intermediate nodes (sets) in the tree.  

Before presenting the three phases, we briefly discuss some of the techniques to be used.

The partition of $D$ in phases one and two follows the two-stage color-coding approach proposed by Bringmann~\cite{Bri17}. The partition in phase one ensures that for any $Z$ with $\sigma(Z) = t$, $|D_i \cap Z| \leq k$ for some $k = \mathrm{polylog}(n,t)$ with high probability. Then in phase two, we randomly partition each $D_i$ into roughly $k^2$ subsets $D_{i,1}, \ldots, D_{i,k^2}$. If any subset $D_{i,j}$ contains at most one element from $D_i \cap Z$, then we have that $\sum_j (D_{i,j} \cup \{0\})$ contains $\sigma(Z \cap D_i)$. Standard balls-into-bins analysis shows that this happens with a probability of at least $1/4$. See Lemma~\ref{lem:standard-color-coding}. The probability can be boosted to $1 - (n+t)^{-\Omega(1)}$ by repeating the process for a logarithmic number of times.

\begin{lemma}\label{lem:standard-color-coding}
    Let $D$ be a multi-set of integers. Let $D_1, \ldots, D_{k^2}$ be a random partition of $D$. We assume that each $D_i$ is a set by removing duplicate integers. For any $Z\subseteq D$ with $|Z| = k$, with probability at least $1/4$, every $D_i$ contains at most one element from $Z$. In other words, for any $Z\subseteq D$ with $|Z| = k$, with probability at least $1/4$, we have
    \[
        \sigma(Z) \in (D_1 \cup \{0\}) + (D_2 \cup \{0\}) + \cdots + (D_{k^2} \cup \{0\}).
    \]
\end{lemma}

Recall that in the dense-or-sparse framework, we merge the results of the groups in a tree-like manner using sparse convolution, and only when a tree level has a small total size do we compute this level. This can be easily done using sparse convolution. Basically, we compute the nodes (sets) in a new level one by one, and stop immediately once we find that the total size of the nodes exceeds a given threshold. See the following lemma for details.

\begin{lemma}\label{lem:dense-or-sparse}
   Let $A_1, \ldots, A_\ell$ be non-empty sets of integers, where $\ell$ is even. Let $u$ be the maximum diameter among all sets $A_i$ (i.e., $u = \max_{i\in[1,\ell]} \dm(A_i)$).  Given any positive integer $k$, in $O((k + u) \,\mathrm{polylog}\,u)$ time, we can either
   \begin{enumerate}[label={\normalfont (\roman*)}] 
        \item compute $\{B_1, \ldots, B_{\ell/2}\}$, or

        \item tell that $\sum_{i=1}^{\ell/2}|B_i| \geq k$,
     \end{enumerate}
    where $B_i = A_{2i-1} + A_{2i}$ for $i \in [1, \ell/2]$.
\end{lemma}
\begin{proof}
    If $k \leq \ell/2$, then $\sum_{i=1}^{\ell/2}|B_i| \geq k$. Assume $k > \ell/2$. For $i = 1, 2, \ldots, \ell/2$, we compute $B_i$ via Lemma~\ref{lem:sparse-fft}. If, at some point, we find that the total size of all sets $B_i$ already exceeds $k$, then we stop immediately. See Algorithm~\ref{alg:sum-or-dense} for details. Let $B_{i^*}$ be the last $B_i$ that is computed by the algorithm.  Since the diameter of $B_{i^*}$ is at most $2u$, its size $|B_{i^*}| \leq 2u + 1$. Therefore, $\sum_{i=1}^{i^*}|B_i| \leq k + 2u + 1$.  Therefore, the running time of the algorithm is $O((k + u) \,\mathrm{polylog}\,u)$.
\end{proof}
\begin{algorithm}[t]
    \caption{\texttt{Sum-If-Sparse}$(\{A_1, \ldots, A_\ell\}, k)$}
    \label{alg:sum-or-dense}
    \begin{algorithmic}[1]
    \Statex \textbf{Input:} non-empty
    subsets $A_1, \ldots , A_\ell$ of integers and an integer $k$.
    \Statex \textbf{Output:} $\{B_1, \ldots, B_{\ell/2}\}$ or tells that $\sum_{i=1}^{\ell/2}|B_i| \geq k$, where $B_i = A_{2i-1} + A_{2i}$
    \State $\mathrm{totalSize} := 0$\;
    \For{$i := 1,..., \ell/2$}
        \State $B_i := A_{2i-1} + A_{2i}$ \Comment{compute via Lemma~\ref{lem:sparse-fft}}\;
        \State $\mathrm{totalSize} := \mathrm{totalSize} + |B_i|$\;
        \If{$\mathrm{totalSize} \geq k$}
            \State Stop immediately\;\Comment{$\sum_{i=1}^{\ell/2}|B_i| \geq k$}
        \EndIf
    \EndFor
    \State \Return $\{B_1, \ldots, B_{\ell/2}\}$\;
    \end{algorithmic}
\end{algorithm}

\subsection{Phase One}

In phase one, we present an algorithm that partitions $D$ into $D_1, \ldots, D_\ell$ for some $\ell \leq 2n$ in a way that, with high probability, each $D_i$ contains at most a logarithmic number of elements from any $Z$ with $\sigma(Z) = t$. For technical reasons, we also need to bound the sum of $\max(D_i)$. (Recall that the max of an empty set is defined to be $0$.)

We first partition $D$ into $1 + \lceil\log w\rceil $ subsets $D^0, D^1, \ldots, $ $D^{\lceil\log w\rceil}$ where 
\[
    D^j = \{x \in D : 2^{j} \leq x < 2^{j+1}\}.
\] 
Define $\alpha_j := \min\{\frac{t}{2^{j-1}}, |D^j|\}$ for $j\in[0,\lceil\log w\rceil]$. Let $Z$ be an arbitrary subset of $D$ with $\sigma(Z) \leq t$.  We immediately have $|Z \cap D^j| \leq \alpha_j$. We further partition each $D^j$ randomly into $\alpha_j$ subsets. We can show that each subset has only a logarithmic number of elements from $Z$. For technical reasons, we require these subsets to be non-empty, and this can be done by rearranging elements to empty subsets. See Algorithm~\ref{alg:first-stage} for details. 

\begin{algorithm}[t]
    \caption{\texttt{Phase-One-Partition}($D, t$)}
    \label{alg:first-stage}
    \begin{algorithmic}[1]
    \Statex \textbf{Input:} a multi-set $D$ of at most $n$ integers from $[1,w]$ and a positive integer $t$
    \Statex \textbf{Output:} A partition $D_1, \ldots, D_\ell$ of $D$  
    \State $D^j := \{x \in D : 2^{j} \leq x < 2^{j+1}\}$ for $j \in [0, \lceil\log w\rceil]$
    \State $\alpha_j := \min\{\frac{t}{2^{j-1}}, |D^j|\}$ for $j \in [0, \lceil\log w\rceil]$
    \For{$j = 0, 1, \ldots, \lceil\log w\rceil$}
        \If{$\alpha_j = |D^j|$}
            \State Partition $D^j$ into subsets $D^j_1, \ldots, D^j_{\alpha_j}$ of size $1$
        \Else
            \State Partition $D^j$ into subsets $D^j_1, \ldots, D^j_{\alpha_j}$ randomly\;\label{in-alg-first-stage:partition-randomly}
            \While{$D^j_i = \emptyset$ and $|D^j_{i'}| \geq 2$ for some $i$ and $i'$}
                \State Move an element from $D^j_{i'}$ to $D^j_i$\;\label{in-alg-first-stage:move-to-empty-set}
            \EndWhile
        \EndIf         
    \EndFor
    \State \Return the collection of all subsets $D^j_i$
    \end{algorithmic}
\end{algorithm}

\begin{lemma}\label{lem:first-stage-color-coding}
    Let $D$ be a multi-set of at most $n$ integers from $[1,w]$. Let $t$ be a positive integer. In $O(n)$ time, we can randomly partition $D$ into subsets $\{D_1, \ldots, D_\ell\}$ such that the following is true.  
    \begin{enumerate}[label={\normalfont (\roman*)}]
        \item $\ell$ is a power of $2$ and $\ell \leq 2n$

        \item for each $j \in [0, \lceil\log w\rceil]$, at most $\frac{t}{2^{j-1}}$ sets $D_i$ have $2^j \leq \max(D_i) < 2^{j+1}$. 

        \item $\frac{3t}{2} \leq \sum_{i=1}^{\ell} \max(D_i) \leq 5t\log w$. 

        \item For any $Z \subseteq D$ with $\sigma(Z) \leq t$, for any $0< q < 1$,
            \[
                \mathbf{Pr}\left[|Z \cap D_i| \leq 6\log (n/q) \text{ for all $i$}\right] \geq 1 - q.
            \]
    \end{enumerate}
\end{lemma}
\begin{proof}
    We first partition $D$ into $D_1, \ldots, D_\ell$ via Algorithm~\ref{alg:first-stage}, and show that properties (ii)--(iv) are satisfied. Property (i) may not hold at this moment since $\ell$ may not be a power of $2$, and we will fix it later.

    Let $\alpha_j, D^j, D^j_i$ be defined as in Algorithm~\ref{alg:first-stage}. 
    Property (ii) is straightforward since each $D^j$ is partitioned into $\alpha_j$ subsets and $\alpha_j \leq \frac{t}{2^{j-1}}$.

    Consider property (iii). We first show that every $D^j_i$ is non-empty. If $\alpha_j = |D^j|$, then each $D^i_j$ contains exactly one element and hence is non-empty. Suppose that $\alpha_j < |D^j|$. Although $D^i_j$ may be empty after the random partition (in line \ref{in-alg-first-stage:partition-randomly}), it is guaranteed to get an element in line \ref{in-alg-first-stage:move-to-empty-set}. The non-emptyness of $D^j_i$ implies $2^j\leq \max(D^j_i) < 2^{j+1}$. Now we are ready to prove (iii). The upper bound is easy as 
    \begin{align*}
        \sum_{j=0}^{\lceil\log w\rceil}\sum_{i=1}^{\alpha_j} \max(D^j_i) &\leq \sum_{j=0}^{\lceil\log w\rceil} 2^{j+1}\cdot \alpha_j \\
        &= \sum_{j=0}^{\lceil\log w\rceil} (4t) \leq 5t\log w.
    \end{align*}
    For the lower bound, it is either the case that $\alpha_j = \frac{t}{2^{j-1}}$ for some $j = j^*$ or $\alpha_j = |D^j|$ for all $j$. In the former case, 
    \[
        \sum_{j=0}^{\lceil\log w\rceil}\sum_{i=1}^{\alpha_j} \max(D^j_i) \geq \sum_{i=1}^{\alpha_{j^*}} \max(D^{j^*}_i) \geq 2^{j^*}\cdot \alpha_{j^*} = 2t.
    \]
    In the latter case, every $D^j_i$ contains exactly one element from $D$, so
    \[
        \sum_{j=0}^{\lceil\log w\rceil}\sum_{i=1}^{\alpha_j} \max(D^j_i) = \sigma(D) \geq \frac{3t}{2}.
    \]
    The last inequality is due to the assumption that $\sigma(D) \geq \frac{3t}{2}$ at the end of Section~\ref{sec:keylem}.

    We next prove property (iv). Consider an arbitrary $D^j_i$. If $\alpha_j = |D^j|$, then the stated probability bound holds since $|D^j_i| = 1$ in this case. Suppose $\alpha_j =\frac{t}{2^{j-1}}$. Fix an arbitrary $Z \subseteq D$ with $\sigma(Z) \leq t$. Let $k = |Z\cap D^j|$. We have that $k \leq \frac{t}{2^j} < \alpha_j$. Consider the $D^j_i$ immediately after the random partition (in line \ref{in-alg-first-stage:partition-randomly}). The set $D^j_i$ is among the $\alpha_j$ subsets that form a random partition of $D$. Therefore, $|Z \cap D^j_i|$ can be viewed as the sum of $k$ independent Bernoulli random variables with success probability $1/\alpha_j$. Since $\mathbf{E}[|Z \cap D^j_i|] = k/\alpha_j < 1 \leq \log \frac{n}{q}$, a standard Chernoff bound gives that 
    \[
                \mathbf{Pr}[|Z \cap D^j_i| > 6\log (n/q)] \leq q/n.
    \]
    Lines \ref{in-alg-first-stage:move-to-empty-set} of the algorithm only increase the size of empty sets $D^j_i$, and hence do not affect the probability bound. Recall that the number of subsets $D^j_i$ is at most $\sum_j \alpha_j \leq n$. By union bound, we have
    \begin{align*}
        &\mathbf{Pr}\left[|Z \cap D_i| \leq 6\log (n/q) \text{ for all $i$}\right] \\
        \geq &1 - \mathbf{Pr}\left[|Z \cap D_i| > 6\log (n/q) \text{ for some $i$}\right] \\
        \geq &1-q.
    \end{align*}

    Now we fix property (i). After Algorithm~\ref{alg:first-stage}, we have that
    \[
        \ell = \sum_{j=0}^{\lceil\log w\rceil} \alpha_j \leq \sum_{j=0}^{\lceil\log w\rceil} |D^j| \leq n.
    \]
    We can force it to be a power of $2$ by adding empty sets $D_i$. It is easy to see that adding empty sets does not affect properties (ii)--(iv) and that $\ell$ increases by at most a factor of $2$.

    The running time of Algorithm~\ref{alg:first-stage} is $O(n)$ since $j$th iteration of the for loop costs $O(|D^j| + \alpha_j)$ time. We add at most $n$ empty sets, which can be done in $O(n)$.
\end{proof}

\subsection{Phase Two}
Let $D_1,\ldots, D_\ell$ be the partition of $D$ given by phase one (Lemma~\ref{lem:first-stage-color-coding}). For any $Z \subseteq D$ with $\sigma(Z) \leq t$, we have $|Z \cap D_i| \leq 6\log n/q$ (with probability at least $1-q$). Let $k = 6\log n/q$.  One can tackle each $D_i$ using standard color-coding: partition $D_i$ into $k^2$ subsets, add $0$ to each subset, and compute the sumset $S_i$ of these subsets in a tree-like manner.  In order to fully utilize the dense-or-sparse framework, we shall tackle all subsets $D_i$ simultaneously. More precisely, instead of computing each $S_i$ independently, we compute all sumsets $S_i$ simultaneously in a forest-like manner. If all levels of the forest have a small total size, then we can compute efficiently using sparse convolution. 
If a level of the forest has a large total size, then we can show that $D$ is dense. 
To reduce the error probability, we will repeat the process for logarithmic times.
See Algorithm~\ref{alg:second-stage} for details.

\begin{algorithm}[t]
    \caption{\texttt{Phase-Two-Partition-and-\\Convolution}$(D_1, \ldots, D_\ell, q)$}
    \label{alg:second-stage}
    \begin{algorithmic}[1]
    \Statex \textbf{Input:} a partition $D_1, \ldots, D_\ell$ of $D$ given by Lemma~\ref{lem:first-stage-color-coding}, and an error probability $q$
    \Statex \textbf{Output:} a set $S_i \subseteq \mathcal{S}(D_i)$ for each $D_i$
    \State $k := 6\log \frac{2n}{q}$
    \State $g := k^2$ rounded up to the next power of $2$ 
    \State $S_i = \emptyset$ for $i \in [1,\ell]$
    \For{$r = 1, \ldots, \lceil\log_{4/3} \frac{4n}{q}\rceil$} 
    \\
    \Comment{repeat the standard color-coding for $\lceil\log_{4/3} \frac{4n}{q}\rceil$ times}
    \For{$i = 1, \ldots, \ell$} 
        \State Randomly partition $D_i$ into subsets $D_{i,1}, \ldots, D_{i,g}$ 
            \State $D^0_{i,j} := D_{i,j}$ for each $j \in [1, g]$ \\  \Comment{kept only for analysis}
            \State $S^0_{i,j} := D_{i,j}\cup \{0\}$ (removing duplicate elements) for each $j \in [1, g]$.
    \EndFor
    \State $u': = \max\{gw+1, 5\sqrt{wt}\log w\}$
    \State $\rho:= 10g\log w$
    \For{$h = 1, \ldots, \log g$} \\\Comment{iteratively compute new levels using Algorithm~\ref{alg:sum-or-dense}}
    \State Run \texttt{Sum-If-Sparse}$(\{S^{h-1}_{1,1},\dots,S^{h-1}_{\ell,g/{2^{h-1}}}\},$ $\frac{\ell g}{2^h} + 4C_{\ap}\rho u'\log u')$ \label{in-alg-second-stage:run sum-if-sparse}
    \If{It returns a result $S'_{1,1},\dots, S'_{\ell,g/{2^h}}$}
    \State $S^h_{i,j}: = S'_{i,j}$ for all $i \in [1, \ell]$ and $j \in [1, \frac{g}{2^h}]$ 
    \\\Comment{$S^h_{i,j}:=S^{h-1}_{i,2j-1} + S^{h-1}_{i,2j}$} 
    \State $D^h_{i,j}: = D^{h-1}_{i,2j-1} \cup D^{h-1}_{i,2j}$ for all $i \in [1, \ell]$ and $j \in [1, \frac{g}{2^h}]$\Comment{kept only for analysis}
    \Else
    \State Stop immediately and \Return \Comment{$D$ is dense}\label{in-alg-second-stage:return-dense}
    \EndIf
    \EndFor
    \State $S_i := S_i \cup S^{\log g}_{i,1}$ for $i \in [1, \ell]$
    \EndFor
    \State \Return $S_1, \ldots, S_{\ell}$
    \end{algorithmic}
\end{algorithm}

We first analyze the running time of Algorithm{\normalfont~\ref{alg:second-stage}}.
\begin{lemma}\label{lem:second-stage-time}
Algorithm{\normalfont~\ref{alg:second-stage}} runs in $O((n+\sqrt{wt}) \polylog$ $(n, t, \frac{1}{q}))$ time.
\end{lemma}
\begin{proof}
    Let all the variables be defined as in Algorithm~\ref{alg:second-stage}. The running time of Algorithm~\ref{alg:second-stage} is dominated by invocations of Algorithm~\ref{alg:sum-or-dense} in line \ref{in-alg-second-stage:run sum-if-sparse}. Since the diameter of every $S^h_{i,j}$ is at most $gw+1$, Lemma~\ref{lem:dense-or-sparse} implies that every invocation of Algorithm~\ref{alg:sum-or-dense} takes ${O}((gw + \ell g + \rho u'\log u')\polylog(gw))$ time. There are at most $\lceil\log_{4/3} \frac{4n}{q}\rceil \cdot \log g$ invocations of Algorithm~\ref{alg:sum-or-dense}. Moreover, Lemma~\ref{lem:first-stage-color-coding} guarantees that $\ell \leq 2n$.  By the selection of the parameters, one can verify that the total running time is $O((n+\sqrt{wt})\polylog(n, t, \frac{1}{q}))$.    
\end{proof}

We prove the correctness of Algorithm{\normalfont~\ref{alg:second-stage}} in the following two lemmas.
\begin{lemma}\label{lem:second-stage-sparse}
    If Algorithm{\normalfont~\ref{alg:second-stage}} returns some sets $S_1,\dots,S_\ell$, then the following holds.
            \begin{enumerate}[label={\normalfont (\roman*)}]
                \item $\max(D_i) \leq \max(S_i) \leq 72\log^2\frac{2n}{q} \max(D_i)$ for any $i\in[1,\ell]$ 

                \item For any $Z \subseteq D$ with $\sigma(Z) \leq t$, 
        \[
            \mathbf{Pr}[\sigma(Z \cap D_i) \in S_i\textrm{ for all $i$}] \geq 1 - q.
        \]
            \end{enumerate}
\end{lemma}
\begin{proof}
We first prove property (i). Each $S_i$ is the union of the sets $S^{\log g}_{i,1}$ in $\lceil\log_{4/3} \frac{4n}{q}\rceil$ repetitions of the standard color-coding. Note that $S^{\log g}_{i,1} = \sum_{j=1}^g S^0_{i,j}$ and that $\max(S^0_{i,j}) \leq \max(D_i)$. Therefore, $\max(S^{\log g}_{i,1}) \leq g\max(D_i) \leq 72\log^2\frac{2n}{q}\max(D_i)$, so is $\max(S_i)$. Since $D_i \subseteq S_i$, it is straightforward that $\max(S_i) \geq \max(D_i)$.

Next, we prove property (ii). Let $Z\subseteq D$ be any subset with $\sigma(Z)\leq t$. Let $Z_i:=Z\cap D_i$ for any $i\in[1,\ell]$. Recall that $D_1, \ldots, D_\ell$ are given by Lemma~\ref{lem:first-stage-color-coding}. Let $k:=6\log\frac{2n}{q}$. Lemma~\ref{lem:first-stage-color-coding} guarantees that 
    \[
                \mathbf{Pr}\left[Z_i \leq k \text{ for all $i$}\right] \geq 1 - \frac{q}{2}.
    \]
    Then Lemma~\ref{lem:standard-color-coding} guarantees that, for each $S^{\log g}_{i,1}$ obtained in a single run of standard color-coding,
    \(
        \mathbf{Pr}[\sigma(Z_i) \in S^{\log g}_{i,1}: |Z_i|\leq k] \geq 1/4.
    \) 
    As the standard color-coding is repeated for  $\lceil\log_{4/3} \frac{4n}{q}\rceil$ times,  we have
    \[
        \mathbf{Pr}[\sigma(Z_i) \in S_i : |Z_i| \leq k] \geq 1 - \frac{q}{4n}.
    \]
    Recall that $\ell \leq 2n$ by Lemma~\ref{lem:standard-color-coding}. By union bound,
    \[
        \mathbf{Pr}[\sigma(Z_i) \in S_i \textrm{ for all $i$} : |Z_i| \leq k \textrm{ for all $i$}] \geq 1 - \frac{q}{2}
    \]
    Therefore,
    \[
            \mathbf{Pr}[\sigma(Z \cap D_i) \in S_i\textrm{ for all $i$}] \geq 1 - q. \qedhere
    \]       
\end{proof}

\begin{lemma}\label{lem:second-stage-dense}
    If Algorithm~\ref{alg:second-stage} stops in line \ref{in-alg-second-stage:return-dense}, then $D$ is dense.
\end{lemma}
\begin{proof}
Suppose that the algorithm stops at line \ref{in-alg-second-stage:return-dense}. That is, Algorithm~\ref{alg:sum-or-dense} finds that $\sum_{i,j}|S^{h'}_{i,j}| \geq \frac{\ell g}{2^{h'}} + 4C_{\ap}\rho u'\log u'$ for some $h'\in[1,\log g]$. We shall show that $D$ satisfies all the conditions in Definition~\ref{def:dense} and hence is dense. Note that $\{D^{h'}_{i,j}\}_{i,j}$ form a partition of $D$ and that $S^{h'}_{i,j} \subseteq \mathcal{S}(D^{h'}_{i,j})$. Condition (i) of Definition~\ref{def:dense} is satisfied. Let $f(S^{h'}_{i,j}) = \max(S^{h'}_{i,j})$. We immediately have that, for all $i$ and $j$,
    \[
        \max(S^{h'}_{i,j}) = f(S^{h'}_{i,j}) \leq \sigma(D^{h'}_{i,j}).
    \] 
    The second inequality is due to that $S^{h'}_{i,j} \subseteq \mathcal{S}(D^{h'}_{i,j})$. Therefore, condition (ii) of Definition~\ref{def:dense} is also satisfied. It remains to verify condition (iii). That is,
    \[
        \frac{3t}{2} \leq \sum_{i,j} \max(S^{h'}_{i,j}) \leq \frac{\rho t}{2} = 5gt\log w.
    \]
    Since $\{D_{i,j}\}_{i,j}$ are obtained by randomly partition each $D_i$ into $g$ subsets, we have
        \[
            \sum_i\max(D_i) \leq \sum_{i,j}\max(D_{i,j})\leq g\sum_i\max(D_i).
        \]
        Recall that $D_1, \ldots, D_\ell$ are given by Lemma~\ref{lem:first-stage-color-coding}. So $\frac{3t}{2} \leq \sum_i \max(D_i) \leq 5t\log w$, which implies  
        \[
            \frac{3t}{2} \leq \sum_{i,j} \max(D_{i,j}) \leq 5gt\log w.
        \]
        It is easy to see that $\sum_{i,j} \max(S^{h+1}_{i,j}) = \sum_{i,j} \max(S^{h}_{i,j})$ for any $h\geq 0$ and $\max(S^0_{i,j}) = \max(D_{i,j})$. Therefore,
        \[
              \frac{3t}{2} \leq \sum_{i,j} \max(S^{h'}_{i,j}) = \sum_{i,j} \max(D_{i,j}) \leq 5gt\log w.   
        \]
        Condition (iii) of Definition~\ref{def:dense} is satisfied.    
\end{proof}

\subsection{Phase Three}
We first conclude the previous two phases by the following lemma.
\begin{lemma}
\label{lem:ult-color-coding}
    Let $D$ be a multi-set of at most $n$ integers from $[1,w]$. Let $t$ be a positive integer. Let $q$ be as $0 < q < 1$. Let $g =72\log^2\frac{2m}{q}$.  In $O((n + \sqrt{wt})\polylog(n, t, \frac{1}{q}))$ time, we can either show that $D$ is dense with respect to $t$ or compute a partition $\{D_1, \ldots, D_\ell\}$ of $D$ and a set $S_i \subseteq \mathcal{S}(D_i)$ for each $D_i$ such that the following holds.
    \begin{enumerate}[label={\normalfont (\roman*)}]
        \item $\ell$ is a power of $2$ and  $\ell \leq 2n$

        \item for each $j \in [0, \lceil\log w\rceil]$, at most $\frac{t}{2^{j-1}}$ sets $D_i$ have $2^j \leq \max(D_i) < 2^{j+1}$. 

        \item $\max(D_i) \leq \max(S_i) \leq g\max(D_i)$ for any $i \in [1, \ell]$

        \item $\frac{3t}{2} \leq \sum_{i=1}^{\ell} \max(S_i) \leq 5gt\log w$ 

        \item For any $Z \subseteq D$ with $\sigma(Z) \leq t$,
        \[
            \mathbf{Pr}[\sigma(Z \cap D_i) \in S_i\textrm{ for all $i$}] \geq 1 - q.
        \]
    \end{enumerate}
\end{lemma}
\begin{proof}
    The first phase partitions $D$ into $D_1, \ldots, D_\ell$ via Lemma~\ref{lem:first-stage-color-coding}. Properties (i) and (ii) are guaranteed by Lemma~\ref{lem:first-stage-color-coding} (i) and (ii). Then the second phase process $D_1, \ldots, D_\ell$ via Algorithm~\ref{alg:second-stage}. If Algorithm~\ref{alg:second-stage} stops in line \ref{in-alg-second-stage:return-dense}, then Lemma~\ref{lem:second-stage-dense} implies that $D$ is dense with respect to $t$, and we are done. Otherwise, Algorithm~\ref{alg:second-stage} returns a set $S_i \subseteq \mathcal{S}(D_i)$ for every $i \in [1, \ell]$. Properties (iii) and (v) are ensured by Lemma~\ref{lem:second-stage-sparse}. Property (iv) is implied by Lemma~\ref{lem:first-stage-color-coding}(iii) and Lemma~\ref{lem:second-stage-sparse}(i).

    The total running time is implied by Lemma~\ref{lem:first-stage-color-coding} and Lemma~\ref{lem:second-stage-time}.
\end{proof}

Let $(D_1, \ldots, D_\ell)$ and $(S_1, \ldots, S_\ell)$ be obtained in the previous two phases. The task of phase three is to compute $S = \sum_i S_i$. There is, however, a trouble with the dense-or-sparse framework. As the sets $S_i$ are merged, their diameters (and the threshold for them to be dense) can be as large as $\Theta(nw)$. As a consequence, we need $\Theta(nw)$ time per level in the tree-like computation, which is too much for our target running time. 

To deal with this trouble, we shall use a random permutation technique, which is inspired by a similar random partition technique for knapsack~\cite{BC23,HX24,BDP24}. The intuition is the following. If we randomly permute all the sets $S_i$ (and their associated sets $D_i$), then with high probability, the contribution of $D_i$ to any $Z \subseteq D$ with $\sigma(Z) \in [t - 5\sqrt{wt}\log w, t]$ is around the mean value. This property holds even when the sets $S_i$ (and $D_i$) are merged.  As a result, in the tree-like computation, we can cap each intermediate node (i.e., sumset) with a short interval around the mean value, and hence reduce the diameters of the nodes.

\subsubsection{Some Probability Bounds}
We first formalize the above intuition.  The following lemma is an application of Bernstein's inequality. (Although the original Bernstein's inequality is only for the sum of independent random variables, a result in Hoeffding's seminal paper~\cite[Theorem 4]{Hoeff63} implies that it also works for sampling without replacement.) 
\begin{restatable}{lemma}{lempartitionandsample}\label{lem:partition-and-sample}
    Let $A$ be a multi-set of $k$ non-negative integers. Let $A_1, \ldots, A_\ell$ be a partition of $A$. Let $B$ be a multi-set of $s$ integers randomly sampled from $A$ without replacement. For any $c \geq 1$, 
    {\begin{align*}
    \mathbf{Pr}\!\left[\!\left|\sigma(B) - \frac{s}{k}\sigma(A)\right| \!>\! 4c \!\cdot\! \log \ell \!\cdot \! \sum_{i}\sqrt{|A_i|}\max(A_i)\right] \!\leq\! e^{-c}.
    \end{align*}}
\end{restatable}

To prove it, we first show that we can bound the probability that $|\sigma(B) - \frac{s}{k}\sigma(A)|$ is large.
\begin{lemma}\label{lem:hoeff}
    Let $A$ be a multi-set of $k$ non-negative integers. Let $B$ be a multi-set of $s$ integers randomly sampled from $A$ without replacement. Then for any $\eta > 0$,
    \begin{align*}
    &\mathbf{Pr}\left[\left|\sigma(B) - \frac{s}{k}\sigma(A)\right| \geq \eta\right] \\\leq &2\mathrm{exp}\left(- \min\left\{\frac{\eta^2}{4\sigma(A)\max(A)}, \frac{\eta}{2\max(A)}\right\}\right) 
    \end{align*}
\end{lemma}
\begin{proof}
    It is easy to see that $\mathbf{E}[\sigma(B)] = s\sigma(A)/k$. We shall use Bernstein's inequality to bound the probability. Although the initial Bernstein's inequality was only for the sum of independent random variables, a result in Hoeffding's seminal paper~\cite[Theorem 4]{Hoeff63} implies that it also works for sampling without replacement (see~\cite[Propostion 1.4]{BM15} for reference). Therefore, by Bernstein's inequality, for any $\eta > 0$,
    \begin{align*}
        &\mathbf{Pr}\left[\left|\sigma(B) - \frac{s}{k}\sigma(A)\right| \geq \eta\right] 
        \\\leq& 2 \mathrm{exp}\left(- \frac{\eta^2}{2\cdot\frac{s}{k}\cdot\sum_{a \in A}\left(a - {\sigma(A)}/{k}\right)^2 + \frac{2}{3}\eta\max(A)}\right) \nonumber\\
        \leq &2\mathrm{exp}\!\left(\!-\! \min\!\left\{\frac{\eta^2}{4\!\cdot\!\frac{s}{k}\!\cdot\!\sum_{a \in A}\!\left(a - {\sigma(A)}/{k}\right)^2}, \frac{\eta}{2\max(A)}\!\right\}\!\right) \label{eq:berstein}.                               
    \end{align*}
    Since that
    \begin{align*}
    \frac{s}{k}\cdot\sum_{a \in A}\left(a - \frac{\sigma(A)}{k}\right)^2  &\leq \sum_{a \in A} \left(a - \frac{\sigma(A)}{k}\right)^2 \\&\leq \sum_{a\in A} a^2 \\&\leq \max(A)\sigma(A),
    \end{align*}
    we obtain the target inequality.
\end{proof}

Now we consider one subset of $A$.

\begin{lemma}\label{lem:sample-div}
    Let $A$ be a multi-set of $k$ non-negative integers. Let $A^*$ be a subset of $A$. Let $B$ be a multi-set of $s$ integers randomly sampled from $A$ without replacement. For any $c \geq 1$, 
    \[
        \mathbf{Pr}\left[\left|\sigma(B\cap A^*) - \frac{s}{k}\sigma(A^*)\right| > 4c\sqrt{|A^*|}\max(A^*)\right] \leq e^{-c}.
    \]
\end{lemma}
\begin{proof}
    When $A^* = \emptyset$, the target inequality trivially holds.  Assume that $A^*$ is non-empty. The integers not in $A^*$ never contribute to $\sigma(B\cap A^*)$, so they can be viewed as $0$. Let $A'$ be the multi-set obtained from $A$ by replacing each integer not in $A^*$ with $0$.  Let $B'$ be a set of $k$ integers randomly sampled from $A'$.  It is easy to see that $\sigma(B' \cap A^*)$ has the same distribution as $\sigma(B \cap A)$. It suffices to show that the target inequality holds for $\sigma(B' \cap A^*)$. By Lemma~\ref{lem:hoeff}, for any $\eta > 0$,
    \begin{align*}
        &\mathbf{Pr}\left[\left|\sigma(B') - \frac{s}{k}\sigma(A')\right| > \eta\right] \\\leq& 2\mathrm{exp}\left(- \min\left\{\frac{\eta^2}{4\sigma(A')\max(A')}, \frac{\eta}{2\max(A')}\right\}\right) 
    \end{align*}
    Since all integers in $A'$ are $0$ except for those in $A^*$, the above inequality is equivalent to 
    \begin{align*}
        &\mathbf{Pr}\left[\left|\sigma(B' \cap A^*) - \frac{s}{k}\sigma(A^*)\right| > \eta\right] \\\leq& 2\mathrm{exp}\left(- \min\left\{\frac{\eta^2}{4\sigma(A^*)\max(A^*)}, \frac{\eta}{2\max(A^*)}\right\}\right).
    \end{align*}
    Let $\eta = 4c\sqrt{|A^*|}\max(A^*)$. One can verify that this probability is bounded by $\mathrm{exp}(-c)$.
\end{proof}

Now we can prove Lemma~\ref{lem:partition-and-sample}.

\begin{proof}[Proof of Lemma~\ref{lem:partition-and-sample}]
    By Lemma~\ref{lem:sample-div}, we have that, for each $A_i$,
    \begin{align*}
    &\mathbf{Pr}\left[\left|\sigma(B \cap A_i) - \frac{s}{k}\sigma(A_i)\right| > 4c \cdot \log \ell \cdot \sqrt{|A_i|}\max(A_i)\right]\\ \leq &e^{-c}/\ell.
    \end{align*}
    Note that $\sigma(B) = \sum_i {\sigma(B \cap A_i)}$ and that $\sigma(A) = \sum_i \sigma(A_i)$. Then
    \[   
    \begin{aligned}
         &\mathbf{Pr}\left[\left|\sigma(B) - \frac{s}{k}\sigma(A)\right| > 4c \cdot \log \ell \cdot  \sum_{i}\sqrt{|A_i|}\max(A_i)\right] \\
         \leq &\sum_i \mathbf{Pr}\!\left[\left|\sigma(B \cap A_i) \!-\! \frac{s}{k}\sigma(A_i)\right| > 4c\! \cdot \!\log \!\ell \!\cdot \!\sqrt{|A_i|}\max(A_i)\right]\\
         \leq &e^{-c}.
    \end{aligned}
    \]
\end{proof}

Now we use Lemma~\ref{lem:partition-and-sample} to prove a probability bound
for our algorithm.

\begin{lemma}\label{lem:small-div}
    Let $(D_1, \ldots, D_\ell)$, $(S_1, \ldots, S_\ell)$, and $q$ be as in Lemma{\normalfont~\ref{lem:ult-color-coding}}. Let $D'$ be the union of $s$ sets $D_i$ that are randomly sampled from $\{D_1, \ldots, D_\ell\}$ without replacement. For any $Z\subseteq D$ such that $\sigma(Z \cap D_i) \in S_i$ for every $i$, we have
    \begin{align*}  \mathbf{Pr}\!\left[\!\left|\sigma(Z\cap D') \!-\! \frac{s}{\ell}\sigma(Z)\right|\! > \!2304\sqrt{wt} \!\cdot \! \log^2\! w \!\cdot  \!\log^3\!\frac{2n}{q}\!\right] \!\leq \!\frac{q}{2n}.  
    \end{align*}
    
\end{lemma}
\begin{proof}
    Let $Z$ be an arbitrary subset of $D$ such that $\sigma(Z \cap D_i) \in S_i$ for every $i$. Define a multi-set $A = \{\sigma(Z \cap D_i) : i \in [1,\ell]\}$. Let $B$ be a multi-set of $s$ numbers randomly sampled from $A$ without replacement.  One can see that $\sigma(Z) = \sigma(A)$ and that $\sigma(B)$ shares the distribution with $\sigma(Z\cap D')$. It suffices to prove the stated probability bound with $\sigma(Z \cap D')$ and $\sigma(Z)$ replaced by $\sigma(B)$ and $\sigma(A)$, respectively.   We partition $A$ into $2 + \lceil\log w\rceil$ subsets $A_0, \ldots, A_{\lceil\log w\rceil + 1}$ where $A_0$ is the multi-set of all $0$'s in $A$, and for $j \geq 1$
    \[
        A_j = \{\sigma(Z \cap D_i) : 2^{j-1} \leq \max(D_i) < 2^{j}\}.
    \]
    By property (ii) of Lemma~\ref{lem:ult-color-coding}, we have $|A_j| \leq \frac{t}{2^{j-2}}$ for every $j \geq 1$. Let $g = 72\log^2\frac{2m}{q}$. Since $\sigma(Z \cap D_i) \in S_i$, by property (iii) of Lemma~\ref{lem:ult-color-coding}, we have $\sigma(Z \cap D_i) \leq g\max(D_i)$, which implies $\max(A_j) < g\cdot 2^{j}$ for every $j \geq 1$.  Therefore, for any $j \in [0, \lceil\log w\rceil + 1]$,
    \[
        \sqrt{|A_j|}\max(A_j) \leq \sqrt{\frac{t}{2^{j-2}}}\cdot g 2^{j} \leq 4g\sqrt{2^{j-w}t} \leq 4g\sqrt{wt}.
    \]
    By Lemma~\ref{lem:partition-and-sample},
    \begin{align*}
\mathbf{Pr}\Bigg[\left|\sigma(B) - \frac{s}{\ell}\sigma(A)\right| > & 4c \cdot \log (1 + \log w) \cdot\\&   \sum_{j}\left(\sqrt{|A_j|}\max(A_j)\right)\Bigg]\leq e^{-c}.
    \end{align*}
    Note that $\sum_{i}(\sqrt{|A_j|}\max(A_j)) \leq 4(1 + \lceil\log w\rceil)g\sqrt{wt}\leq 8g\sqrt{wt}\log w$. Let $c = \log \frac{2n}{q}$. We have
     \[
         \mathbf{Pr}\left[\left|\sigma(B) - \frac{s}{\ell}\sigma(A)\right| > 32g\cdot\sqrt{wt} \cdot  \log^2 w \cdot  \log\frac{2n}{q}\right] \leq \frac{q}{2n}.
    \]
    Recall that $g = 72\log^2\frac{2n}{q}$. This completes the proof.
\end{proof}

\subsubsection{The Algorithm}
Recall that, in the sparse case, we only care about $Z \subseteq D$ with $\sigma(Z) \in [t - 5\sqrt{wt}\log w, t]$. Lemma~\ref{lem:small-div} implies that for an intermediate node (sumset), we can cap it using an interval of length roughly $\widetilde{O}(\sqrt{wt})$. After capping, each intermediate node (sumset) has a small diameter, and we can use the dense-or-sparse framework.  See Algorithm~\ref{alg:permute-and-conv} for details.

\begin{algorithm}[t]
    \caption{\texttt{Random-Permutation-and-Sparse-\\Convolution}}
    \label{alg:permute-and-conv}
    \begin{algorithmic}[1]
    \Statex \textbf{Input:} $S_1, \ldots, S_{\ell}$, $D_1, \ldots, D_{\ell}$ and $q$ be as in Lemma~{\normalfont~\ref{lem:ult-color-coding}}
    \Statex \textbf{Output:} a set $S \subseteq  \mathcal{S}(D)$
    \State Permute $\{S_1, \ldots, S_{\ell}\}$ (and the associated $D_1, \ldots,$ $D_\ell$) randomly
    \State $S^0_i := S_i$ and $D^0_i := D_i$ for each $i \in [1,\ell]$
    \State $\eta: = 2304\sqrt{wt} \cdot  \log^2 w \cdot  \log^3\frac{2n}{q} + 5\sqrt{wt}\log w$
    \State $u' :=4\eta + 1$
    \State $\rho := 10g\log w$
    \For{$h := 1, \ldots, \log \ell$}
\State $\ell_h := \frac{\ell}{2^h}$\;
\State Run \texttt{Sum-If-Sparse}$(\{S^{h-1}_{1},\dots,S^{h-1}_{2\ell_{h}}\},\ell_h + 4\cdot C_{\ap}\rho u'\log u')$\label{in-alg-permute-and-conv:run-Sum-If-Sparse}
    \If{It returns a result $S'_{1},\dots, S'_{\ell_h}$}
    \State $S^h_{i}: = S'_{i}$ for all $i \in [1, \ell_h]$ 
    \Comment{$S^h_{i}:=S^{h-1}_{2i-1} + S^{h-1}_{2i}$} 
    \State $S^h_i := S^h_i \cap [\frac{t}{\ell_h} - \eta, \frac{t}{\ell_h} + \eta]$ for each $i \in [1, \ell_h]$\label{in-alg-permute-and-conv:cap}
    \State $D^h_i := D^{h-1}_{2i-1} \cup D^{h-1}_{2i}$ for $i \in [1, \ell_h]$\; \\\Comment{kept only for analysis}
    \Else
    \State Stop immediately and \Return \Comment{$D$ is dense}\label{in-alg-permute-and-conv:return-dense}
    \EndIf
    \EndFor
    \State \Return $S^{\log {\ell}}_1$\;
    \end{algorithmic}
\end{algorithm}

We first analyze the running time of Algorithm{\normalfont~\ref{alg:permute-and-conv}}.
\begin{lemma}\label{lem:permute-and-conv-time}
    Algorithm{\normalfont~\ref{alg:permute-and-conv}} runs in $O((n + \sqrt{wt})\cdot\polylog$ $ (n, t, \frac{1}{q}))$ time.
\end{lemma}
\begin{proof}
    Let the variables be as in Algorithm~\ref{alg:permute-and-conv}. The first 5 lines cost $O(n)$ time. The running time of the rest is dominated by the $\log \ell \leq 1 + \log n$ invocations of Algorithm~\ref{alg:sum-or-dense} in line \ref{in-alg-permute-and-conv:run-Sum-If-Sparse}, and each invocation costs $O(\eta + \ell_h + \rho u'\log u')\polylog \eta$ time.  One can verify that the total running time is $O((n + \sqrt{wt})\polylog (n, t, \frac{1}{q}))$.
\end{proof}

Next, we prove the correctness of Algorithm{\normalfont~\ref{alg:permute-and-conv}}.
\begin{lemma}\label{lem:permute-and-conv-sparse}
    If Algorithm{\normalfont~\ref{alg:permute-and-conv}} returns a set $S$, then $|S| = O(\sqrt{wt}\polylog (n, t, \frac{1}{q}))$, and for any $Z \subseteq D$ with $\sigma(Z) \in [t - 5\sqrt{wt}\log w, t]$, 
    \[
        \mathbf{Pr}[\sigma(Z) \in S] \geq 1 - 3q.
    \]
\end{lemma}
\begin{proof}
    Let the variables be defined as in Algorithm~\ref{alg:permute-and-conv}. Let $Z$ be an arbitrary subset of $D$ with $\sigma(Z) \in [t - 5\sqrt{wt}\log w, t]$. Assume that $\sigma(Z \cap D_i) \in S_i$ for all $i$. Lemma~\ref{lem:ult-color-coding} guarantees that the error probability of this assumption is at most $q$. Now consider an arbitrary $D^h_i$. Due to the random permutation of $\{S_1, \ldots, S_{\ell}\}$ (and the associated $D_1, \ldots, D_\ell$), the set $D^h_i$ can be view as the union of $2^h$ sets $D_i$ that are randomly sampled from $D_1, \ldots, D_\ell$. By Lemma~\ref{lem:small-div}, we have 
    \[
        \mathbf{Pr}\left[\left|\sigma(Z\!\cap\! D^h_i) \!-\! \frac{1}{\ell_h}\!\sigma(Z)\right|\! > \!32g\!\cdot\!\sqrt{wt} \!\cdot\!  \log^2\! w \!\cdot\!  \log\frac{2n}{q}\right] \!\leq\! \frac{q}{2n}.  
    \]
    Since $\sigma(Z) \in [t - 5\sqrt{wt}\log w, t]$, for $\eta = 2304\sqrt{wt} \cdot  \log^2 w \cdot  \log^3\frac{2n}{q} + 5\sqrt{wt}\log w$, we have
    \[
        \mathbf{Pr}\left[\left|\sigma(Z\cap D^h_i) - \frac{t}{\ell_h}\right| > \eta \right] \leq \frac{q}{2n}.  
    \]
    Note that there are at most $2\ell \leq 4n$ sets $D^h_i$. We assume that $\sigma(Z\cap D^h_i) \in [\frac{t}{\ell_h} - \eta, \frac{t}{\ell_h} + \eta]$ for all $D^h_i$. The error probability of this assumption is at most $\frac{q}{2n} \cdot 2\ell = 2q$. Since the cap operation in line \ref{in-alg-permute-and-conv:cap} of the algorithm does not lose $\sigma(Z\cap D^h_i)$, we have $\sigma(Z) \in S^{\log \ell}_1$. The total error probability is $3q$.

    Moreover, $|S^{\log \ell}_1| \leq 2\eta = O(\sqrt{wt}\polylog (n, t, \frac{1}{q}))$.
\end{proof}

\begin{lemma}\label{lem:permute-and-conv-dense}
    If Algorithm~\ref{alg:permute-and-conv} stops in line \ref{in-alg-permute-and-conv:return-dense}, then $D$ is dense with respect to $t$.
\end{lemma}
\begin{proof}
    Suppose that the algorithm stops when $h=h'$. Let the variables be defined as in Algorithm~\ref{alg:permute-and-conv}. We show that $D$ satisfies all the conditions of Definition~\ref{def:dense}. The sets
    $S^{h'}_1, \ldots, S^{h'}_{\ell_{h'}}$ already satisfy the condition (i) of Definition~\ref{def:dense}. We still need to show that there is a function $f$ satisfying conditions (ii) and (iii) of Definition~\ref{def:dense}.
    We define a function $f$ recursively as follows.
    \begin{align*}
        &f(S^0_i) = \max(S_i)\\
        &f(S^{h+1}_i) = f(S^h_{2i-1}) + f(S^h_{2i})
    \end{align*}
    To satisfy conditions (ii) and (iii) of Definition~\ref{def:dense}, it suffices to show that $\frac{3t}{2} \leq \sum_{i=1}^{\ell_h} f(S^{h}_i) \leq 5gt\log w$ and $\sigma(D^{h}_i) \geq f(S^{h}_i)$ for any $h$. We prove by induction on $h$. When $h = 0$, we have $f(S^{0}_i) = \max(S_i)$. Since $S_1, \ldots, S_\ell$ is given by Lemma~\ref{lem:ult-color-coding}, we have
    \[
        \frac{3t}{2} \leq \sum_{i=1}^\ell \max(S_i) \leq 5gt\log w.
    \]
    Moreover, $\max(S_i) \leq \sigma(D_i)$ since $S_i \subseteq \mathcal{S}(D_i)$. So the two inequalities hold for $h = 0$. 
    Suppose that the two inequalities hold for $h$. We show that it holds for $h+1$. By the definition of $f$, we have $\sum_{i=1}^{\ell_{h+1}} f(S^{h+1}_i) = \sum_{i=1}^{\ell_h} f(S^h_i)$. By the inductive hypothesis,
    \[
        \frac{3t}{2} \leq \sum_{i=1}^{\ell_{h+1}} f(S^{h+1}_i) \leq 5gt\log w.
    \]
    Moreover, $D^{h+1}_i = D^h_{2i-1} \cup D^h_{2i}$ and $f(S^{h+1}_i) = f(S^h_{2i-1}) + f(S^h_{2i})$. By the inductive hypothesis, we also have 
    \[
        f(S^{h+1}_i) \leq \sigma(D^{h+1}_i). \qedhere
    \]
\end{proof}

\subsection{Putting Things Together}
We conclude the sparse case by proving Lemma~\ref{lem:sparse}. The key lemma (Lemma~\ref{lem:key}) will follow straightforwardly by Lemma~\ref{lem:good-ap} and Lemma~\ref{lem:sparse}.
\lemsparse*
\begin{proof}
    We first process $D$ via Lemma~\ref{lem:ult-color-coding} with $q = (n + t)^{-\Omega(1)}$. The time cost is $\widetilde{O}(n + \sqrt{wt})$.  If it shows that $D$ is dense with respect to $t$, then we are done. Otherwise, it returns a partition $D_1, \ldots, D_\ell$ of $D$ and a set $S_i \subseteq \mathcal{S}(D_i)$ for each $D_i$. 
    We process them via Algorithm~\ref{alg:permute-and-conv}. By Lemma~\ref{lem:permute-and-conv-time}, the time cost is again $\widetilde{O}(n + \sqrt{wt})$. If it shows that $D$ is dense, then we are done. Otherwise, by Lemma~\ref{lem:permute-and-conv-dense}, it returns a set $S$ of $\widetilde{O}(\sqrt{wt})$ size such that for any $Z \subseteq D$ with $\sigma(Z) \in [t - 5\sqrt{wt}\log w, t]$, 
    \[
        \mathbf{Pr}[\sigma(Z) \in S] \geq 1 - (n+t)^{-\Omega(1)}.
    \]
Then we can cap $S$ as $S: = S \cap [t-5\sqrt{wt}\log w, t]$. We have that $S \subseteq \mathcal{S}(D) \cap [t - 5\sqrt{wt}\log w, t]$ and that $S$ contains all $s \in \mathcal{S}(D) \cap [t - 5\sqrt{wt}\log w, t]$ with probability $1 - (n + t)^{-\Omega(1)}$.
\end{proof}

\section{Conclusion}
We obtain a randomized $\widetilde{O}(n+\sqrt{wt})$-time algorithm for Subset Sum, which improves the $\widetilde{O}(n+t)$-time algorithm by Bringmann~\cite{Bri17} in the regime $t \geq w$.

An important open question in this line of research is whether Subset Sum can be solved in $\widetilde{O}(n+w)$ time. It is worth mentioning that when the input is a set, we can obtain an $\widetilde{O}(n+w^{1.25})$-time algorithm. However, when the input is a multi-set, in the regime of $t\gg w^2$ the $\widetilde{O}(n+w^{1.5})$-time algorithm by Chen, Lian, Mao and Zhang~\cite{CLMZ24aSODA} remains the best so far. It is interesting if one can get an algorithm for Subset Sum in $\widetilde{O}(n+w^{1.5-\varepsilon})$-time for a fixed constant $\epsilon>0$. 

\backmatter



\section*{Declarations}

\bmhead{Funding} 

This work was supported by National Natural Science Foundation of China (Project No. [6257070197], [62402436], and [12271477]).

\bmhead{Conflict of interest} 

The authors have no competing interests to declare that are relevant to the content of this article.



\bibliography{sn-bibliography}

@misc{Jin23,
      title="Solving {{Knapsack}} with small items via L0-proximity", 
      author={Ce Jin},
      year={2023},
      eprint={2307.09454},
      archivePrefix={arXiv},
      primaryClass={cs.DS},
      url={https://arxiv.org/abs/2307.09454}, 
}

@inproceedings{GGC20,
author = {Giorgi, Pascal and Grenet, Bruno and Cray, Armelle Perret du},
title = {Essentially optimal sparse polynomial multiplication},
year = {2020},
organization = {Association for Computing Machinery},
doi = {10.1145/3373207.3404026},
booktitle = {Proceedings of the 45th International Symposium on Symbolic and Algebraic Computation (ISSAC 2020)},
pages = {202–209}
}

@inproceedings{Pot21,
  author =	{Pot\k{e}pa, Krzysztof},
  title =	{{Faster Deterministic Modular Subset Sum}},
  booktitle =	{29th Annual European Symposium on Algorithms (ESA 2021)},
  pages =	{7601-7616},
  year =	{2021},
  organization =	{Schloss Dagstuhl -- Leibniz-Zentrum f{\"u}r Informatik},
  doi =		{10.4230/LIPIcs.ESA.2021.76}
}

@inproceedings{CI21,
author = {Jean Cardinal and John Iacono},
title = {Modular Subset Sum, Dynamic Strings, and Zero-Sum Sets},
booktitle = {2021 Symposium on Simplicity in Algorithms (SOSA 2021)},
year =	{2021},
pages = {45-56},
doi = {10.1137/1.9781611976496.5},
organization =	{Society for Industrial and Applied Mathematics}
}

@inproceedings{CMZ26,
author = {Lin Chen and Yuchen Mao and Guochuan Zhang},
title = {Long Arithmetic Progressions in Sparse Subset Sums: A Computational Perspective},
booktitle = {Proceedings of the 2026 Annual ACM-SIAM Symposium on Discrete Algorithms (SODA 2026)},
year =	{2026},
pages = {3611-3638},
doi = {10.1137/1.9781611978971.132},
organization =	{Society for Industrial and Applied Mathematics}
}

@inproceedings{ABB+,
author = {Kyriakos Axiotis and Arturs Backurs and Karl Bringmann and Ce Jin and Vasileios Nakos and Christos Tzamos and Hongxun Wu},
title = {Fast and Simple Modular Subset Sum},
booktitle = {2021 Symposium on Simplicity in Algorithms (SOSA 2021)},
year =	{2021},
pages = {57-67},
doi = {10.1137/1.9781611976496.6},
organization =	{Society for Industrial and Applied Mathematics}
}

@inproceedings{ABJ+19,
  title={Fast {{Modular Subset Sum}} using linear sketching},
  author={Axiotis, Kyriakos and Backurs, Arturs and Jin, Ce and Tzamos, Christos and Wu, Hongxun},
  booktitle={Proceedings of the Thirtieth Annual ACM-SIAM Symposium on Discrete Algorithms (SODA 2019)},
  pages={58--69},
  year={2019},
  organization={Society for Industrial and Applied Mathematics},
  doi = {10.1137/1.9781611975482.4}
}

@article{ABHS22,
    author = "Abboud, Amir and Bringmann, Karl and Hermelin, Danny and Shabtay, Dvir",
    title = "{{SETH}}-based Lower Bounds for {{Subset Sum}} and {{Bicriteria Path}}",
    year = "2022",
    month = "January",
    fulljournal = "ACM Transactions on Algorithms",
    journal = "ACM Trans. Algorithms",
    volume = "18",
    number = "1",
    pages = "1--22",
    issn = "1549-6325, 1549-6333",
    doi = "10.1145/3450524"
}

@article{Alo87,
    author = "Alon, N.",
    title = "Subset Sums",
    year = "1987",
    month = "October",
    fulljournal = "Journal of Number Theory",
    journal = "J. Number Theory",
    volume = "27",
    number = "2",
    pages = "196--205",
    issn = "0022-314X",
    doi = "10.1016/0022-314X(87)90061-8"
}

@inproceedings{AR15,
    author = "Arnold, Andrew and Roche, Daniel S.",
    title = "Output-Sensitive Algorithms for {{Sumset}} and Sparse Polynomial Multiplication",
    booktitle = "Proceedings of the 2015 {{ACM}} on International Symposium on {{Symbolic}} and Algebraic Computation ({{ISSAC}} 2015)",
    organization = {Association for Computing Machinery},
    year = "2015",
    pages = "29--36",
    doi = "10.1145/2755996.2756653",
    isbn = "978-1-4503-3435-8"
}

@inproceedings{AT19,
    author = "Axiotis, Kyriakos and Tzamos, Christos",
    title = "Capacitated Dynamic Programming: Faster {{Knapsack}} and Graph Algorithms",
    booktitle = "46th International Colloquium on {{Automata}}, {{Languages}}, and {{Programming}} ({{ICALP}} 2019)",
    organization =	{Schloss Dagstuhl -- Leibniz-Zentrum f{\"u}r Informatik},
    year = "2019",
    pages = "19:1--19:13",
    doi = "10.4230/LIPIcs.ICALP.2019.19",
    isbn = "978-3-95977-109-2",
    issn = "1868-8969"
}

@inproceedings{BC22,
    author = "Bringmann, Karl and Cassis, Alejandro",
    title = "Faster {{Knapsack}} Algorithms via Bounded Monotone Min-Plus-Convolution",
    booktitle = "49th International Colloquium on Automata, Languages, and Programming (ICALP 2022)",
    organization =	{Schloss Dagstuhl -- Leibniz-Zentrum f{\"u}r Informatik},
    year = "2022",
    pages = "31:1--31:21",
    doi = "10.4230/LIPIcs.ICALP.2022.31",
    isbn = "978-3-95977-235-8",
    issn = "1868-8969"
}

@inproceedings{BC23,
    author = "Bringmann, Karl and Cassis, Alejandro",
    title = "Faster 0-1-{{Knapsack}} via Near-Convex Min-Plus-Convolution",
    booktitle = "31st Annual European Symposium on Algorithms (ESA 2023)",
    organization =	{Schloss Dagstuhl -- Leibniz-Zentrum f{\"u}r Informatik},
    year = "2023",
    pages = "2401--2416", 
    doi = "10.4230/LIPIcs.ESA.2023.24",
    isbn = "978-3-95977-295-2",
    issn = "1868-8969"
}

@inproceedings{BDP24,
    author = {Bringmann, Karl and D\"{u}rr, Anita and Polak, Adam},
    title =	"Even Faster {{Knapsack}} via Rectangular Monotone Min-Plus Convolution and Balancing",
    booktitle =	{32nd Annual European Symposium on Algorithms (ESA 2024)},
    organization =	{Schloss Dagstuhl -- Leibniz-Zentrum f{\"u}r Informatik},
    year = {2024},
    pages = {3301--3315}, 
    doi = {10.4230/LIPICS.ESA.2024.33},
    ISBN =	{978-3-95977-338-6},
    ISSN =	{1868-8969}
}

@book{Bel57,
    author = "Bellman, Richard",
    title = "Dynamic {{Programming}}",
    year = "1957",
    volume = "33",
    publisher = "{Princeton University Press}",
    address = {Princeton},
    doi = "10.2307/j.ctv1nxcw0f",
    isbn = "978-0-691-14668-3"
}

@inproceedings{BFN21,
    author = "Bringmann, Karl and Fischer, Nick and Nakos, Vasileios",
    title = "Sparse Nonnegative Convolution Is Equivalent to Dense Nonnegative Convolution",
    booktitle = "Proceedings of the 53rd Annual ACM SIGACT Symposium on Theory of Computing (STOC 2021)",
    organization = {Association for Computing Machinery},
    year = "2021",
    pages = "1711--1724",
    doi = "10.1145/3406325.3451090",
    isbn = "978-1-4503-8053-9"
}

@inproceedings{BFN22,
    author = "Bringmann, Karl and Fischer, Nick and Nakos, Vasileios",
    title = "Deterministic and {{Las Vegas}} Algorithms for Sparse Nonnegative Convolution",
    booktitle = "Proceedings of the 2022 Annual ACM-SIAM Symposium on Discrete Algorithms (SODA 2022)",
    organization = {Society for Industrial and Applied Mathematics},
    year = "2022",
    pages = "3069--3090",
    doi = "10.1137/1.9781611977073.119",
    isbn = "978-1-61197-707-3"
}

@inproceedings{BHSS18,
    author = "Bateni, MohammadHossein and Hajiaghayi, MohammadTaghi and Seddighin, Saeed and Stein, Cliff",
    title = "Fast Algorithms for {{Knapsack}} via Convolution and Prediction",
    booktitle = "Proceedings of the 50th Annual ACM SIGACT Symposium on Theory of Computing (STOC 2018)",
    organization = {Association for Computing Machinery},
    year = "2018",
    pages = "1269--1282",
    doi = "10.1145/3188745.3188876",
    isbn = "978-1-4503-5559-9"
}

@article{BM15,
    author = {R{\'e}mi Bardenet and Odalric-Ambrym Maillard},
    title = {{Concentration inequalities for sampling without replacement}},
    volume = {21},
    fulljournal = {Bernoulli},
    journal = {Bernoulli},
    number = {3},
    publisher = {Bernoulli Society for Mathematical Statistics and Probability},
    pages = {1361 -- 1385},
    year = {2015},
    doi = {10.3150/14-BEJ605},
    URL = {https://doi.org/10.3150/14-BEJ605}
}

@inproceedings{BN20,
    author = "Bringmann, Karl and Nakos, Vasileios",
    title = "Top-k-Convolution and the Quest for near-Linear Output-Sensitive {{Subset Sum}}",
    booktitle = "Proceedings of the 52nd Annual ACM SIGACT Symposium on {{Theory}} of {{Computing}} ({{STOC}} 2020)",
    organization = {Association for Computing Machinery},
    year = "2020",
    pages = "982--995",
    doi = "10.1145/3357713.3384308",
    isbn = "978-1-4503-6979-4"
}

@inproceedings{BN21a,
    author = "Bringmann, Karl and Nakos, Vasileios",
    title = "Fast N-Fold Boolean Convolution via Additive Combinatorics",
    booktitle = "48th International Colloquium on Automata, Languages, and Programming (ICALP 2021)",
    organization =	{Schloss Dagstuhl -- Leibniz-Zentrum f{\"u}r Informatik},
    year = "2021",
    pages = "41:1--41:17",
    doi = "10.4230/LIPIcs.ICALP.2021.41",
    isbn = "978-3-95977-195-5",
    issn = "1868-8969"
}

@inproceedings{BN21b,
    author = "Bringmann, Karl and Nakos, Vasileios",
    title = "A Fine-Grained Perspective on Approximating {{Subset Sum}} and {{Partition}}",
    booktitle = "Proceedings of the 2021 ACM-SIAM Symposium on Discrete Algorithms ({{SODA}} 2021)",
    organization = {Society for Industrial and Applied Mathematics},
    year = "2021",
    pages = "1797--1815",
    doi = "10.1137/1.9781611976465.108",
    isbn = "978-1-61197-646-5"
}

@inproceedings{Bri17,
    author = "Bringmann, Karl",
    title = "A Near-Linear Pseudopolynomial Time Algorithm for {{Subset Sum}}",
    booktitle = "Proceedings of the 2017 Annual ACM-SIAM Symposium on Discrete Algorithms ({{SODA}} 2017)",
    organization = {Society for Industrial and Applied Mathematics},
    year = "2017",
    pages = "1073--1084",
    doi = "10.1137/1.9781611974782.69",
    isbn = "978-1-61197-478-2"
}

@inproceedings{Bri24,
    author = {Bringmann, Karl},
    title = {Knapsack with Small Items in Near-Quadratic Time},
    booktitle = {Proceedings of the 56th Annual ACM Symposium on Theory of Computing (STOC 2024)},
    organization = {Association for Computing Machinery},
    year = {2024},
    pages = {259--270},
    doi = {10.1145/3618260.3649719},
    isbn = {979-8-40070-383-6}
}

@inproceedings{BW21,
    author = "Bringmann, Karl and Wellnitz, Philip",
    title = "On Near-Linear-Time Algorithms for Dense {{Subset Sum}}",
    booktitle = "Proceedings of the 2021 ACM-SIAM Symposium on Discrete Algorithms (SODA 2021)",
    organization = {Society for Industrial and Applied Mathematics},
    year = "2021",
    pages = "1777--1796",
    doi = "10.1137/1.9781611976465.107",
    isbn = "978-1-61197-646-5"
}

@article{CFG89,
    author = "Chaimovich, Mark and Freiman, Gregory and Galil, Zvi",
    title = "Solving Dense Subset-Sum Problems by Using Analytical Number Theory",
    year = "1989",
    month = "September",
    fulljournal = "Journal of Complexity",
    journal = "J. Complexity",
    volume = "5",
    number = "3",
    pages = "271--282",
    issn = "0885-064X",
    doi = "10.1016/0885-064X(89)90025-3"
}

@inproceedings{CH02,
    author = "Cole, Richard and Hariharan, Ramesh",
    title = "Verifying Candidate Matches in Sparse and Wildcard Matching",
    booktitle = "Proceedings of the Thiry-Fourth Annual {{ACM}} Symposium on {{Theory}} of Computing ({{STOC}} 2002)",
    organization = {Association for Computing Machinery},
    year = "2002",
    pages = "592--601",
    doi = "10.1145/509907.509992",
    isbn = "978-1-58113-495-7"
}

@article{Cha99,
    title={New algorithm for dense subset-sum problem},
    author={Chaimovich, Mark},
    fulljournal={Ast{\'e}risque},
    journal={Ast{\'e}risque},
    volume={258},
    pages={363--373},
    year={1999},
    publisher = {Soci\'et\'e math\'ematique de France},
}

@inproceedings{CL15,
    author = "Chan, Timothy M. and Lewenstein, Moshe",
    title = "Clustered {{Integer 3SUM}} via Additive Combinatorics",
    booktitle = "Proceedings of the Forty-Seventh Annual {{ACM}} Symposium on {{Theory}} of {{Computing}} ({{STOC}} 2015)",
    organization = {Association for Computing Machinery},
    year = "2015",
    pages = "31--40",
    doi = "10.1145/2746539.2746568",
    isbn = "978-1-4503-3536-2"
}

@inproceedings{CLMZ24aSODA,
    author = {Lin Chen and Jiayi Lian and Yuchen Mao and Guochuan Zhang},
    title = {Faster Algorithms for Bounded {{Knapsack}} and Bounded {{Subset Sum}} Via Fine-Grained Proximity Results},
    booktitle = {Proceedings of the 2024 Annual ACM-SIAM Symposium on Discrete Algorithms (SODA 2024)},
    organization = {Society for Industrial and Applied Mathematics},
    year = "2024",
    pages = {4828-4848},
    doi = {10.1137/1.9781611977912.171},
    isbn = "978-1-61197-791-2"
}

@inproceedings{CLMZ24cSTOCPartition,
    author = {Chen, Lin and Lian, Jiayi and Mao, Yuchen and Zhang, Guochuan},
    title = {Approximating {{Partition}} in Near-Linear Time},
    booktitle = {Proceedings of the 56th Annual ACM Symposium on Theory of Computing (STOC 2024)},
    organization = {Association for Computing Machinery},
    year = {2024},
    pages = {307--318},
    doi = {10.1145/3618260.3649727},
    isbn = {9798400703836}
}

@inproceedings{DJM23,
    author = "Deng, Mingyang and Jin, Ce and Mao, Xiao",
    title = "Approximating {{Knapsack}} and {{Partition}} via Dense {{Subset Sums}}",
    booktitle = "Proceedings of the 2023 Annual ACM-SIAM Symposium on Discrete Algorithms (SODA 2023)",
    organization = {Society for Industrial and Applied Mathematics},
    year = "2023",
    pages = "2961--2979",
    doi = "10.1137/1.9781611977554.ch113",
    isbn = "978-1-61197-755-4"
}

@article{GM91,
    author = "Galil, Zvi and Margalit, Oded",
    title = "An Almost Linear-Time Algorithm for the Dense Subset-Sum Problem",
    year = "1991",
    month = "December",
    fulljournal = "SIAM fulljournal on Computing",
    journal = "SIAM J. Comput.",
    volume = "20",
    number = "6",
    pages = "1157--1189",
    issn = "0097-5397",
    doi = "10.1137/0220072"
}

@article{Hoeff63,
    title={Probability Inequalities for Sums of Bounded Random Variables},
    author={Hoeffding, Wassily},
    fulljournal={Journal of the American Statistical Association},
    journal={J. Amer. Statist. Assoc.},
    volume={58},
    number={301},
    pages={13--30},
    year={1963},
    publisher={Taylor \& Francis},
    doi = {10.1080/01621459.1963.10500830}
}

@inproceedings{HX24,
    author = {Qizheng He and Zhean Xu},
    title = {Simple and Faster Algorithms for {{Knapsack}}},
    booktitle = {Proceedings of 2024 Symposium on Simplicity in Algorithms (SOSA 2024)},
    year ={2024},
    pages = {56-62},
    doi = {10.1137/1.9781611977936.6},
}

@article{IK75,
    author = "Ibarra, Oscar H. and Kim, Chul E.",
    title = "Fast Approximation Algorithms for the Knapsack and Sum of Subset Problems",
    year = "1975",
    month = "October",
    fulljournal = "Journal of the ACM",
    journal = "J. ACM",
    volume = "22",
    number = "4",
    pages = "463--468",
    issn = "0004-5411",
    doi = "10.1145/321906.321909"
}

@inproceedings{Jin24,
    author = "Jin, Ce",
    title = "0-1 {{Knapsack}} in Nearly Quadratic Time",
    booktitle = {Proceedings of the 56th Annual ACM Symposium on Theory of Computing (STOC 2024)},
    organization = {Association for Computing Machinery},
    pages = {271--282},
    year = "2024",
    doi = {10.1145/3618260.3649618},
    isbn = {9798400703836}
}

@article{JR23,
    title={On integer programming, discrepancy, and convolution},
    author={Jansen, Klaus and Rohwedder, Lars},
    fulljournal={Mathematics of Operations Research},
    journal={Math. Oper. Res.},
    volume={48},
    number={3},
    pages={1481--1495},
    year={2023},
    publisher={INFORMS},
    doi={10.1287/moor.2022.1308}
}

@inproceedings{JW19,
    author = "Jin, Ce and Wu, Hongxun",
    title = "A Simple Near-Linear Pseudopolynomial Time Randomized Algorithm for {{Subset Sum}}",
    booktitle = "2nd {{Symposium}} on {{Simplicity}} in {{Algorithms}} ({{SOSA}} 2019)",
    organization =	{Schloss Dagstuhl -- Leibniz-Zentrum f{\"u}r Informatik},
    year = "2019",
    pages = "17:1--17:6",
    issn = "2190-6807",
    doi = "10.4230/OASIcs.SOSA.2019.17",
    isbn = "978-3-95977-099-6"
}

@incollection{Kar72,
    author="Karp, Richard M.",
    editor="Miller, Raymond E. and Thatcher, James W. and Bohlinger, Jean D.",
    title="Reducibility among Combinatorial Problems",
    booktitle="Complexity of Computer Computations",
    year="1972",
    organization="Springer US",
    pages="85--103",
    isbn="978-1-4684-2001-2",
    doi="10.1007/978-1-4684-2001-2_9"
}

@inproceedings{Kar75,
    author = "Karp, Richard M",
    title = "The fast approximate solution of hard combinatorial problems",
    booktitle = "Proc. 6th South-Eastern Conf. Combinatorics, Graph Theory and Computing",
    organization="Utilitas Mathematica Publishing Co.",
    pages = "15--31",
    year = "1975"
}

@inproceedings{Kle22,
title={On the Fine-Grained Complexity of the {{Unbounded SubsetSum}} and the {{Frobenius}} Problem},
author={Klein, Kim-Manuel},
booktitle={Proceedings of the 2022 Annual ACM-SIAM Symposium on Discrete Algorithms (SODA 2022)},
doi = {10.1137/1.9781611977073.141},
pages={3567--3582},
year={2022},
organization={Society for Industrial and Applied Mathematics}
}

@inproceedings{KPS97,
    author = "Kellerer, Hans and Pferschy, Ulrich and Speranza, Maria Grazia",
    title = "An Efficient Approximation Scheme for the Subset-Sum Problem",
    booktitle = "Algorithms and {{Computation}}",
    organization = "Springer",
    year = "1997",
    series = "Lecture {{Notes}} in {{Computer Science}}",
    pages = "394--403",
    doi = "10.1007/3-540-63890-3\_42",
    isbn = "978-3-540-69662-9"
}

@misc{KX18,
    author = "Koiliaris, Konstantinos and Xu, Chao",
    title = "{{Subset Sum}} made simple",
    year = "2018",
    month = "July",
    number = "arXiv:1807.08248",
    eprint = "1807.08248",
    primaryclass = "cs",
    archiveprefix = "arxiv",
    note = "Preprint at \url{https://arxiv.org/abs/1807.08248}"
}

@article{KX19,
    author = "Koiliaris, Konstantinos and Xu, Chao",
    title = "Faster Pseudopolynomial Time Algorithms for {{Subset Sum}}",
    year = "2019",
    month = "July",
    fulljournal = "ACM Transactions on Algorithms",
    journal = "ACM Trans. Algorithms",
    volume = "15",
    number = "3",
    pages = "1--20",
    issn = "1549-6325, 1549-6333",
    doi = "10.1145/3329863"
}

@inproceedings{MWW19,
    author = "Mucha, Marcin and W{\k{e}}grzycki, Karol and W{\l}odarczyk, Micha{\l}",
    title = "A Subquadratic Approximation Scheme for {{Partition}}",
    booktitle = "Proceedings of the 2019 Annual ACM-SIAM Symposium on Discrete Algorithms ({{SODA}} 2019)",
    organization = {Society for Industrial and Applied Mathematics},
    year = "2019",
    pages = "70--88",
    doi = "10.1137/1.9781611975482.5",
    isbn = "978-1-61197-548-2"
}

@article{Nak20,
    author = "Nakos, Vasileios",
    title = "Nearly Optimal Sparse Polynomial Multiplication",
    year = "2020",
    month = "November",
    fulljournal = "IEEE Transactions on Information Theory",
    journal = "IEEE Trans. Inform. Theory",
    volume = "66",
    number = "11",
    pages = "7231--7236",
    issn = "0018-9448, 1557-9654",
    doi = "10.1109/TIT.2020.2989385"
}

@article{Pis99,
    author = "Pisinger, David",
    title = "Linear Time Algorithms for Knapsack Problems with Bounded Weights",
    year = "1999",
    month = "October",
    fulljournal = "Journal of Algorithms",
    journal = "J. Algorithms",
    volume = "33",
    number = "1",
    pages = "1--14",
    issn = "01966774",
    doi = "10.1006/jagm.1999.1034"
}

@article{Pis03,
  title="Dynamic Programming on the Word {{RAM}}",
  author={Pisinger},
  fulljournal={Algorithmica},
  journal = "Algorithmica",
  volume={35},
  pages={128--145},
  year={2003},
  publisher={Springer},
  doi={10.1007/s00453-002-0989-y}
}

@inproceedings{PRW21,
    author = "Polak, Adam and Rohwedder, Lars and W{\k{e}}grzycki, Karol",
    title = "Knapsack and {{Subset Sum}} with Small Items",
    booktitle = "48th International Colloquium on Automata, Languages, and Programming (ICALP 2021)",
    organization =	{Schloss Dagstuhl -- Leibniz-Zentrum f{\"u}r Informatik},
    year = "2021",
    pages = "106:1--106:19",
    doi = "10.4230/LIPIcs.ICALP.2021.106",
    isbn = "978-3-95977-195-5",
    issn = "1868-8969"
}

@article{SV05,
    author = "Szemer{\'e}di, E. and Vu, V.",
    title = "Long Arithmetic Progressions in Sumsets: {{Thresholds}} and Bounds",
    year = "2005",
    month = "September",
    fulljournal = "Journal of the American Mathematical Society",
    journal = "J. Amer. Math. Soc.",
    volume = "19",
    number = "1",
    pages = "119--169",
    issn = "0894-0347, 1088-6834",
    doi = "10.1090/S0894-0347-05-00502-3"
}

@article{SV06,
    author = "Szemer{\'e}di, E. and Vu, V. H.",
    title = "Finite and Infinite Arithmetic Progressions in Sumsets",
    year = "2006",
    month = "January",
    fulljournal = "Annals of Mathematics",
    journal = "Ann. of Math.",
    volume = "163",
    number = "1",
    eprinttype = "jstor",
    pages = "1--35",
    issn = "0003-486X",
    url = "https://www.jstor.org/stable/20159950"
}

@article{Sar89,
    author = {S{\'a}rk{\"o}zy, A.},
    title = "Finite Addition Theorems, {{I}}",
    year = "1989",
    month = "May",
    fulljournal = "Journal of Number Theory",
    journal = "J. Number Theory",
    volume = "32",
    number = "1",
    pages = "114--130",
    issn = "0022-314X",
    doi = "10.1016/0022-314X(89)90102-9"
}

@article{Sar94,
    author = {S{\'a}rk{\"o}zy, A.},
    title = "Fine Addition Theorems, {{II}}",
    year = "1994",
    month = "August",
    fulljournal = "Journal of Number Theory",
    journal = "J. Number Theory",
    volume = "48",
    number = "2",
    pages = "197--218",
    issn = "0022-314X",
    doi = "10.1006/jnth.1994.1062"
}

@article{Tam09,
    author = "Tamir, Arie",
    title = "New Pseudopolynomial Complexity Bounds for the Bounded and Other Integer {{Knapsack}} Related Problems",
    year = "2009",
    month = "September",
    fulljournal = "Operations Research Letters",
    journal = "Oper. Res. Lett.",
    volume = "37",
    number = "5",
    pages = "303--306",
    issn = "01676377",
    doi = "10.1016/j.orl.2009.05.003"
}

@misc{WC22,
    author = "Wu, Xiaoyu and Chen, Lin",
    title = "Improved Approximation Schemes for {{(Un-)bounded Subset-Sum}} and {{Partition}}",
    year = "2022",
    month = "December",
    number = "arXiv:2212.02883",
    eprint = "2212.02883",
    archiveprefix = "arxiv",
    url={https://arxiv.org/abs/2212.02883}, 
}

\end{document}